\documentclass[sn-mathphys-num]{sn-jnl}% Math and Physical Sciences Numbered Reference Style
\usepackage[T1]{fontenc}
\usepackage{amssymb,amsmath,latexsym,mathabx}
\usepackage{amsthm}
\usepackage{mathrsfs}
\usepackage{diagbox}
\usepackage[shortlabels]{enumitem}
\usepackage{color}
\usepackage{comment}
\usepackage{anyfontsize}
\usepackage{xcolor}
\usepackage{stackengine}
\usepackage{calc}
\usepackage{physics}
\usepackage{accents}
\usepackage{hyperref}
\usepackage{appendix}
\usepackage{colortbl}
\usepackage{booktabs}
\usepackage{tabularx}
\usepackage[most]{tcolorbox}
\usepackage{hyperref}
\hypersetup{
  colorlinks   = true,
  urlcolor     = blue,
  linkcolor    = blue,
  citecolor    = blue
}
\allowdisplaybreaks
\newcommand{\wt}{\widetilde}

\newcommand{\mb}{\mathbb}
\newcommand{\mc}{\mathcal}
\newcommand{\supp}{\mathrm{supp}}

\renewcommand{\ketbra}[2]{\mathinner{|{#1}\rangle\,\langle{#2}|}}

\newcommand{\added}[1]{#1}

\newtheorem{theorem}{Theorem}[section]
\newtheorem{definition}[theorem]{Definition}
\newtheorem{remark}[theorem]{Remark}

\newtheorem{lemma}[theorem]{Lemma}
\newtheorem{prop}[theorem]{Proposition}
\newtheorem{example}[theorem]{Example}

\newtheorem{cor}[theorem]{Corollary}
\DeclareMathOperator{\id}{id}

\begin{document}
\title[Reverse-type Data Processing Inequality]{Reverse-type Data Processing Inequality}

\author[1,2]{Paula Belzig}
\author[3]{Li Gao}
\author[1,2]{Graeme Smith}
\author[1,2,*]{Peixue Wu\footnote[0]{*Contact author: p33wu@uwaterloo.ca}}
\affil[1]{Department of Applied Mathematics, University of Waterloo, Waterloo, Ontario, Canada}
\affil[2]{Institute for Quantum Computing, University of Waterloo, Waterloo, Ontario, Canada}
\affil[3]{School of Mathematics and Statistics, Wuhan University, Wuhan, Hubei, China}

\abstract{
The quantum data processing inequality asserts that two quantum states become harder to distinguish when a noisy channel is applied. On the other hand, a reverse quantum data processing inequality characterizes whether distinguishability is preserved after the application of a noisy channel. In this work, we explore these concepts through contraction and expansion coefficients of the relative entropy of quantum channels. Our first result is that quantum channels with an input dimension greater than or equal to the output dimension do not have a non-zero expansion coefficient, which means that they cannot admit a reverse data-processing inequality. We propose a comparative approach by introducing a relative expansion coefficient, to assess how one channel expands relative entropy compared to another. We show that this relative expansion coefficient is positive for three important classes of quantum channels: depolarizing channels, generalized dephasing channels, and amplitude damping channels. As an application, we give the first rigorous construction of level-1 less noisy quantum channels that are non-degradable.

}


\maketitle
\tableofcontents

\section{Introduction}

\iffalse
\noindent

The noisy nature of a quantum channel is reflected in the fact that measures of distance in quantum information decrease when the channel is applied. This property of a distance measure is referred to as monotonicity under quantum operations, or data processing inequality.
One notable example of quantum information measure with wide applications is the relative entropy, that for two quantum states $\rho$ and $\sigma$,
\begin{equation*}
D(\rho||\sigma)= \begin{cases}
 \trace\Big(\rho\big(\log(\rho)-\log(\sigma)\big)\Big), &\text{ if } \text{supp}(\rho)\subseteq \text{supp}(\sigma), \\
\infty, &\text{else,}
\end{cases}
\end{equation*}
which quantifies the distinguishability of two quantum states $\rho$ and $\sigma$ in the context of quantum hypothesis testing \cite{HP91,ON00}. The data processing inequality for quantum relative entropy implies that two states can only become harder to distinguish after a quantum channel is applied. Formally, it states that, for any quantum channel $\mc N\in \mc L(\mb B(\mc H_A), \mb B(\mc H_B))$, we have that
\begin{align*}
     D(\mc N(\rho)\|\mc N(\sigma)) \leq  D( \rho\| \sigma),\ \forall \rho, \sigma.
\end{align*}
This was first proven in \cite{Uhlmann77} and is now a fundamental tool in quantum information processing.
\fi

\noindent
%The noisy nature of a quantum channel can be reflected in tshe contraction of various information measures. 
%Such monotonicity  of information measures under quantum channels are often called data processing inequality.
The noisy nature of a quantum channel is reflected in the fact that measures of distance in quantum information decrease when the channel is applied. This property of an information measure is referred to as monotonicity under quantum operations, or data processing inequality (DPI).
One notable quantum information measure with wide applications is the relative entropy of two quantum states $\rho$ and $\sigma$,
\begin{equation*}
D(\rho||\sigma)= \begin{cases}
 \trace\Big(\rho\big(\log(\rho)-\log(\sigma)\big)\Big) &\text{ if } \text{supp}(\rho)\subseteq \text{supp}(\sigma), \\
\infty &\text{ else}.
\end{cases}
\end{equation*}
The DPI for quantum relative entropy states that for any quantum channel $\mc N$, 
\begin{align*}
     D(\mc N(\rho)\|\mc N(\sigma)) \leq  D( \rho\| \sigma),\ \forall \rho, \sigma.
\end{align*}
Since the relative entropy quantifies how well a quantum state $\rho$ can be distinguished from $\sigma$ in the context of quantum hypothesis testing (see e.g. \cite{HP91,ON00}), the DPI implies that two quantum states can only become less distinguishable after a channel is applied. The DPI of quantum relative entropy was first proven in \cite{lindblad1975} (see \cite{Uhlmann77, Petz03, Muller_Hermes_2017, HT24} for later alterative proofs) and is now a fundamental tool in quantum information processing.

Given a channel $\mc N$, its contraction coefficient \cite{LR99,HR15,HRS22,H2024,caputo2024entropycontractionsmarkovchains} quantifies to what extent the DPI can be improved for this channel. The contraction coefficient is defined as the smallest constant $\eta_{\mc N}$ such that,
\begin{align*}
     D(\mc N(\rho)\|\mc N(\sigma)) \leq \eta_{\mc N} D( \rho\| \sigma),\quad \forall \rho, \sigma
\end{align*}
or equivalently 
\begin{equation}\label{def:contraction coefficient}
   \eta_{\mc N}:= \sup_{\rho \neq \sigma} \frac{D(\mc N(\rho) || \mc N(\sigma))}{D(\rho || \sigma)}.
\end{equation}
The contraction coefficient $\eta_{\mc N}$ characterizes how much harder it becomes to distinguish quantum states after the channel $\mc N$ is applied. By DPI, it is clear that $\eta_{\mc N}\in [0,1]$ for any channel $\mc N$. If $\eta_{\mc N}<1$, the distinguishability of $\rho$ and $\sigma$ decays exponentially fast under the repeated applications of the same channel. In this case, we say that $\mc N$ obeys \textit{strong data processing inequality} (SDPI).
%For channels that obey SDPI, every pair of quantum states becomes harder to distinguish and repeated applications of the same channel will make any two input states completely indistinguishable in the asymptotic limit.

On the other hand, one can also ask whether a channel $\mc N$ must necessarily destroy all distinguishability in the worst case, or if it preserves at least a fixed fraction of information. To capture this, the \emph{expansion coefficient} was introduced in \cite{RISB21,LS23,H2024}:
\begin{equation} \label{def:expansion coefficient}
    \widecheck{\eta}_{\mc N}:= \inf_{\rho \neq \sigma} \frac{D(\mc N(\rho) || \mc N(\sigma))}{D(\rho || \sigma)} \in [0,1].
\end{equation}
If the expansion coefficient of a channel is strictly greater than 0, it must necessarily preserve some information about the states. This can be interpreted as a reverse data processing inequality for $\mc N$ (in direct analogy with reverse Doeblin coefficients \cite{H2024}, reverse Pinsker inequality \cite{Sason15,SV16}, and reverse log-Sobolev inequality \cite{Sontz99}). Operationally, since relative entropy governs exponential error‐decay rates in quantum hypothesis testing (via Stein’s lemma, see~\cite{hiai1991proper}), \(\widecheck\eta_{\mathcal{N}}\) measures the channel’s worst‐case ability to preserve information.
 
\added{Moreover, a positive expansion coefficient guarantees exponential convergence of iterative algorithms such as quantum Blahut–Arimoto method for channel‐capacity computation \cite{RISB21}. In this work, we show that requiring \(\widecheck\eta_{\mathcal{N}}>0\) imposes strong dimension‐dependent constraints, and we prove that beyond certain input‐output size thresholds, one must have \(\widecheck\eta_{\mathcal{N}}=0\): }
\medskip
\begin{theorem}[c.f. Theorem \ref{main:impossibility}]
  For a quantum channel $\mc N:\mb B(\mc H_A)\to \mb B(\mc H_B)$ with $d_A \ge d_B$, we have  %$\text{dim} \mc H_A \ge \text{dim} \mc H_B$, then 
    \begin{equation}
    \widecheck{\eta}_{\mc N}=% \inf_{\rho \neq \sigma} \frac{D(\mc N(\rho) || \mc N(\sigma))}{D(\rho||\sigma)} = 
    \begin{cases}
        1,\quad d_A =d_B \text{ and } \ \mc N(\rho) = U\rho U^{\dagger} \text{ for some unitary U}, \\
        0,\quad \text{otherwise}.
    \end{cases}
    \end{equation}
\end{theorem}
\noindent For channels with greater output than input dimension, the same can not be true because it is easy to construct flagged channels such as erasure channels with $\widecheck{\eta}_{\mc N}>0$.%Here the assumption $d_A\ge d_B$ on the dimension is necessary because if the output dimension is greater than the input dimension it is easy to construct flagged channels such as erasure channels with $\widecheck{\eta}_{\mc N}>0$.

The above result suggests that the expansion coefficient $\widecheck{\eta}_{\mc N}$ does not serve as a reliable standalone measure of information preservation. 
In this work, we propose a comparative approach, where two quantum channels $\mc N$ and $\mc M$ are analyzed based on how they contract or expand the relative entropy \emph{relative to each other}. To formalize this, we introduce the \emph{relative expansion coefficient} $\widecheck{\eta}_{\mc N,\mc M}$ and \emph{relative contraction coefficient} $\eta_{\mc N,\mc M}$:
\begin{align}
 \widecheck{\eta}_{\mc N,\mc M} := \inf_{\rho \neq \sigma, \supp(\rho) \subseteq \supp(\sigma)} \frac{D(\mc N(\rho) \| \mc N(\sigma))}{D(\mc M(\rho) \| \mc M(\sigma))}\ , \  \ \   \eta_{\mc N,\mc M} := \sup_{\rho \neq \sigma,\supp(\rho) \subseteq \supp(\sigma)} \frac{D(\mc N(\rho) \| \mc N(\sigma))}{D(\mc M(\rho) \| \mc M(\sigma))}. \label{def:rel-expan}
\end{align}These two are essentially the same definition by noting that $\widecheck{\eta}_{\mc N,\mc M}=\eta_{\mc M,\mc N}^{-1}$.
The relative contraction coefficient $\eta_{\mc N,\mc M}$ is also referred to as the \emph{less noisy domination factor} in \cite{MP18, HRS22}.

We present systematic tools for analyzing the relative expansion of two quantum channels, including a BKM metric comparison (Section \ref{sec:BKM}) and a complete positivity comparison (Section \ref{sec:CP}). We also provide an sufficient condition for qubit channels such that the relative expansion coefficients $\widecheck{\eta}_{\mc N,\mc M}$ is positive (Section \ref{sec:qubit calculation}). Based on those methods, we investigate the cases when $\mc N$ and $\mc M$ are a pair of depolarizing channels, generalized dephasing channels, and amplitude damping channels respectively. These are three most important classes of quantum channels studied in the literature \cite{preskill1998lecture}. Our results show that the relative expansion coefficient $\widecheck{\eta}_{\mc N,\mc M}$ is often positive (non-trivial) when $\mc N$ and $\mc M$ are related by the degrading condition that $\mc D \circ \mc M = \mc N$ for some quantum channel $\mc D$, establishing a reverse-type data processing inequality for the following cases: 
\medskip
\begin{theorem} We have the following estimates of relative contraction and expansion coefficients. (see Section \ref{sec: examples} for details)
\begin{table}[h]
\centering
\renewcommand{\arraystretch}{1.5} % Adjust row height (increase spacing between rows)
\label{tab1}
\begin{tabularx}{\textwidth}{>{\raggedright\arraybackslash}X % Left-aligned column
                        >{\centering\arraybackslash}X       % Centered column
                        >{\centering\arraybackslash}X}      % Centered column
\toprule
\textbf{Channels} $(\mathcal{N}, \mathcal{M})$ &  $\eta_{\mathcal{N}, \mathcal{M}}$ &  $\widecheck{\eta}_{\mathcal{N}, \mathcal{M}}$ \\
\midrule
\textbf{d-dimension depolarizing}: $(\mathcal{D}_{p_1}, \mathcal{D}_{p_2})$ & $ \le \left( \dfrac{1-p_1}{1-p_2} \right)^2 \dfrac{1 - \frac{d-1}{d}p_2}{1 - \frac{d-1}{d}p_1}$ & $ \ge \left( \dfrac{1-p_1}{1-p_2} \right)^2 \dfrac{p_2}{p_1}$ \\
\addlinespace
\textbf{Qubit depolarizing}: $(\mathcal{D}_{p_1}, \mathcal{D}_{p_2})$ & $=\left( \dfrac{1-p_1}{1-p_2} \right)^2$ & $= \left( \dfrac{1-p_1}{1-p_2} \right)^2 \dfrac{p_2(2-p_2)}{p_1(2-p_1)}$ \\
\addlinespace
\textbf{Generalized dephasing}: $(\Phi_{\Gamma}, \Phi_{\Gamma'})$ & $= 1$ & $> 0$ if $\Gamma$ and $\Gamma'$ are close \\
\addlinespace
\textbf{Qubit dephasing}: $(\Phi_{p}, \Phi_{p'})$ & $= 1$ & $> 0$ \\
\addlinespace
\textbf{Qubit amplitude damping}: $(\mathcal{A}_{\gamma_1}, \mathcal{A}_{\gamma_2})$ & $\le \sqrt{\dfrac{1- \gamma_1}{1 - \gamma_2}}$ & $> 0$ \\
\bottomrule
\end{tabularx}
\end{table}
\end{theorem}
As an application, we utilize the framework of contraction coefficients and relative expansion coefficients to construct quantum channels that are (level-1) less noisy but not degradable. Roughly speaking, a quantum channel is considered (level-1) less noisy if the information in the output system is not less than the information contained in the environment system when we allow the system to couple with an arbitrary classical system (see Section~\ref{sec: preliminaries} for the rigorous definition). Our construction is as follows:
\medskip
\begin{theorem}[c.f. Proposition \ref{prop:ampdamp regions}]
Suppose $\mc A_{\gamma}$ is the amplitude damping channel defined in \eqref{def:amplitude damping channel}. The quantum channel $$\Psi_{p,\gamma_1,\gamma_2}(\rho) = p \ketbra{0}{0}\otimes \mc A_{\gamma_1}(\rho) + (1-p) \ketbra{1}{1} \otimes \mc A_{\gamma_2}(\rho)$$ is less noisy if 
\vspace{5pt}
\begin{itemize}
\item $\gamma_1 + \gamma_2 >1$ and $\gamma_1 < \frac{1}{2}$, and 
\begin{align*}
    \displaystyle p \in \left[ \frac{1}{1 + \widecheck{\eta}_{\mc A_{\gamma_1}, \mc A_{1-\gamma_2}} (1- \eta_{\mc A_{\wt \gamma_1}})} , 1\right],\quad \wt \gamma_1 = \frac{1-2\gamma_1}{1-\gamma_1}.
\end{align*}
\item $\gamma_1 + \gamma_2 >1$ and $\gamma_2 < \frac{1}{2}$, and 
\begin{align*}
    p \in \left[0, \frac{\widecheck{\eta}_{\mc A_{\gamma_2}, \mc A_{1-\gamma_1}} (1- \eta_{\mc A_{\wt \gamma_2}})}{1 + \widecheck{\eta}_{\mc A_{\gamma_2}, \mc A_{1-\gamma_1}} (1- \eta_{\mc A_{\wt \gamma_2}})}\right],\quad \wt \gamma_2 = \frac{1-2\gamma_2}{1-\gamma_2}.
\end{align*}
\end{itemize}
\end{theorem}
A concrete example of a less noisy but not degradable channel is $\Psi_{p,\gamma_1,\gamma_2}$ for parameters $p=0.75$, $\gamma_1=0.2$ and $\gamma_2=0.81$. In fact, we obtain a whole parameter region, which we illustrate in Figure~\ref{fig:rainbow}. 

Our motivation stems from a central problem in quantum information theory: determining the capacities of various quantum channels. While capacities are of fundamental importance, their computation is notoriously intractable because it often requires infinite regularization, or tensorization~\cite{Lloyd_1997, shor2002, devetak2005private, cubitt2015unbounded}. Degradable channels were introduced in \cite{devetak2005capacity} as the first class of channels whose quantum capacity does not require regularization, making them computable via optimization. Since then, it has been shown that even weaker conditions than degradability can preserve additivity, allowing capacities to be computed through optimization \cite{Watanabe12, Cross_2017, SW24}.

Recently, a hierarchy of ``less noisy'' channel classes was introduced in \cite{HRS22}. However, establishing clear separations between these classes remains an open problem. In fact, no known example demonstrates a channel that is ``less noisy'' but fails to be degradable, even under the weakest notion of ``less noisy''.

We address this problem by providing the first explicit construction of a less noisy but non-degradable channel that can be rigorously verified. Since the class of less noisy channels coincides with those having concave coherent information, our construction also confirms the existence of non-degradable channels with concave coherent information. The tools we introduce also provide a potential way to show the existence of a non-degradable channel that is informationally degradable as introduced in \cite{Cross_2017}. This is another characteristic of a quantum channel that lies in between less noisy and degradable and implies additivity of capacity. %, see Section \ref{sec:conclusion and open problems} for more about this problem. 
A key insight is that tensorizing a channel with the identity on an ancillary system performs better than repeatedly tensoring the channel with itself, and this perspective underlies our proposed extensions. 

As a summary, we highlight our main results in the following three aspects: 
\begin{enumerate}
    \item \textbf{Expansion coefficients}. We give a systematic study of the relative entropy expansion coefficient corresponding to a reverse data processing inequality. The contraction coefficient, often called strong data processing constant, has been well-studied in both classical and quantum setting over decades. The expansion coefficient, in contrast, has not been much considered in the literature. We fill this gap by giving the first systematic treatment of the expansion coefficient. 
    \item \textbf{No-go Theorem~\ref{main:impossibility}}. Under the assumption that the input dimension of the channel is not less than the output dimension, we prove the expansion coefficient with respect to relative entropy must be zero. This leads us to study the more meaningful relative expansion coefficients, and obtain non-trivial relative expansion coefficients between  depolarizing, dephasing, and amplitude-damping channels. This is in sharp contrast to the expansion coefficient with respect to the trace distance~\cite{chiribella2025maximum}.
    \item \textbf{Less-noisy but non-degradable channels}. We give a concrete and rigorous construction of less noisy but non-degradable channel using the relative expansion coefficient. To our knowledge, this is the first concrete and rigorous example of such a separation, filling a gap in the literature.
%    \item \textit{Clarification on tensorization}. Tensorized (multi-copy) inequalities are crucial for many operational tasks, yet are notoriously subtle in the quantum setting. Our numerical experiments for amplitude-damping channels suggest the optimal constant in the \emph{tensorized} reverse data–processing inequality equals the single-copy constant. A proof would yield the first example of an informationally degradable but non-degradable channel. 
\end{enumerate}

\medskip

%This also provides evidence to the existence of a non-degradable example of informationally degradable channels, as introduced in \cite{Cross_2017}.  Please see Section \ref{sec:conclusion and open problems} for more about this problem.

The rest of this manuscript is organized as follows. In Section \ref{sec: preliminaries}, we briefly review necessary preliminaries on quantum channels and degradability. In Section \ref{sec: impossibility of reverse data processing}, we prove that if the dimension of the input system of a quantum channel is not less than that of the output system, then the expansion coefficient is zero. Section \ref{sec: relative contraction and expansion} presents systematic tools for analyzing the relative expansion coefficients of two quantum channels. Section \ref{sec: examples} is devoted to explicit estimates of relative contraction and expansion coefficients of three important classes of quantum channels. We then present the construction of non-degradable channels that are less noisy in Section \ref{sec: applications}. Section \ref{sec:conclusion and open problems} concludes the manuscript with a discussion on open problems.

\section{Preliminaries}\label{sec: preliminaries}
\subsection{Quantum channel and its representation}
In this work, we denote $\mc H$ as a Hilbert space of finite dimension, and $\mc H^{\dagger}$ as the dual space of $\mc H$. $\ket{\psi}$ denotes a vector in $\mc H$ and $\bra{\psi} \in \mc H^{\dagger}$ a dual vector. For two Hilbert spaces $\mc H_A, \mc H_B$, the space of linear operators from $\mc H_A$ to $\mc H_B$ is denoted as $\mb B(\mc H_A, \mc H_B) \cong \mc H_B \otimes \mc H_A^{\dagger}$. When $\mc H_A = \mc H_B = \mc H$, we write $\mb B(\mc H, \mc H)$ shortly as $\mb B(\mc H)$. The set of density operators (positive semidefinite with trace one) on $\mc H$ is denoted as $\mc D(\mc H)$. The set of pure states (rank 1 projections) on $\mc H$ is denoted as $\mc P(\mc H)$. Denote $\mc L(\mb B(\mc H_A), \mb B(\mc H_B))$ as the class of super-operators which consists of linear maps from $\mb B(\mc H_A)$ to $\mb B(\mc H_B)$. A quantum channel $\mc N\in \mc L(\mb B(\mc H_A), \mb B(\mc H_B))$ is a super-operator which is completely positive and trace-preserving (CPTP). 

Let $\mc H_A,\mc H_B,\mc H_E$ be three Hilbert spaces of dimensions $d_A,d_B,d_E$ respectively.
 An isometry $V: \mc H_A \to \mc H_B \otimes \mc H_E$, meaning $V^{\dagger}V = I_{A}$ (identity operator on $\mc H_A$), generates a pair of quantum channels $(\mc N, \mc N^c)$, defined by 
\begin{equation}
    \mc N(\rho) = \tr_{E}(V \rho V^{\dag}), \quad \mc N^c(\rho) = \tr_{B}(V \rho V^{\dag}),
\end{equation}
where $\tr_E$ is the partial trace operator given by $\tr_E(X_B \otimes X_E) = \tr(X_E) X_B$. It is known from Stinespring theorem that every quantum channel $\mc N$ can be expressed as above, and 
the pair $(\mc N, \mc N^c)$ is called the \textit{complementary channel} of the other. 

%The operator-sum representation of a quantum channel is given by $B_i = A_i^{\dagger}$, and in this case, we call it \textit{Kraus representation}:
The operator-sum representation of a quantum channel is called \textit{Kraus representation}:
\begin{equation}\label{eqn:Kraus representation}
    \mc N(X) = \sum_{i=1}^m A_i X A_i^{\dagger}, \quad X \in \mb B(\mc H_A),
\end{equation}
where $A_i \in \mb B(\mc H_A, \mc H_B)$ are called Kraus operators of $\mc N$.
Another representation of a super-operator in $\mc L(\mb B(\mc H_A), \mb B(\mc H_B))$ is its Choi–Jamio\l{}kowski operator. Given an orthonormal basis $\{\ket{i}\}_{i=0}^{d_A -1}$  of $\mc H_A$, a maximally entangled state on $\mc H_A \otimes \mc H_A$ is given by 
\begin{align*}
    \ket{\Phi} = \frac{1}{\sqrt{d_A}} \sum_{i=0}^{d_A -1} \ket{i} \otimes \ket{i}.
\end{align*}
The \textit{(unnormalized) Choi–Jamio\l{}kowski} operator of $\mc N \in \mc L(\mb B(\mc H_A), \mb B(\mc H_B))$ is a bipartite operator in $\mb B(\mc H_A \otimes \mc H_B)$ given by
\begin{equation}\label{eqn:Choi operator}
    \mc C_{\mc N} = d_A (\id_{\mb B(\mc H_A)} \otimes \mc N) (\ketbra{\Phi}{\Phi})= \sum_{i,j = 0}^{d_A - 1} \ketbra{i}{j} \otimes \mc N(\ketbra{i}{j}).
\end{equation}
%Note that it is well-known that 
A quantum channel $\mc N$ is completely positive if and only if its Choi–Jamio\l{}kowski operator $\mc C_{\mc N}$ is a positive operator in $B(\mc H_A \otimes \mc H_B)$, and $\mc N$ is trace-preserving if and only if $\tr_B(\mc C_{\mc N}) = I_A$. The rank of $\mc C_{\mc N}$ is called the \textit{Kraus rank} of the channel $\mc N$, which indicates the minimum number of Kraus operators to represent $\mc N$ in \eqref{eqn:Kraus representation}.

For two completely positive superoperator $\mc M$ and $\mc N$, we say $\mc N \le_{cp} \mc M $ if $\mc M-\mc N$ is completely positive. This is equivalent to
\begin{equation}\label{eqn: cp order}
    \mc C_{\mc N}\le \mc C_{\mc M} ,
\end{equation}
where $\mc C_{\mc N}\le \mc C_{\mc M}$ means $\mc C_{\mc M} - \mc C_{\mc N}$ is positive semidefinite.

\subsection{Degradable and less noisy channels}Let $\mc N$ be a quantum channel and $\mc N^c$ be its complementary channel.
We say that $\mc N$ is \textit{degradable} if there is a quantum channel $\mc D$ such that $\mc D \circ\mc N = \mc N^c$. That is, one can process the output system to get all the information about the environment system. Similarly, if there exists a quantum channel $\wt{\mc D}$ such that $\wt{\mc D} \circ\mc N^c = \mc N$, then we say that $\mc N$ is \textit{anti-degradable}. 

Given any additional quantum system $\mc H_V$ and a bipartite density operator $\rho_{VA}$ on $V\otimes A$, denote 
\begin{align*}
    \rho_{VB} = (\id_{\mb B(\mc H_V)} \otimes \mc N) (\rho_{VA}),\quad \rho_{VE} = (\id_{\mb B(\mc H_V)} \otimes \mc N^c) (\rho_{VA}).
\end{align*}
We say $\mc N$ is \textit{informationally degradable} (introduced in \cite{Cross_2017}) if for any quantum system $V$ and bipartite density operator $\rho_{VA}$, we have 
\begin{align*}
    I(V;B)_{(\id_{\mb B(\mc H_V)}\otimes \mc N) (\rho_{VA})} \ge I(V;E)_{(\id_{\mb B(\mc H_V)} \otimes \mc N^c)(\rho_{VA})} ,
\end{align*}
where $I(V;B) = S(V) + S(B) - S(VB)$ is the quantum mutual information and $S(V)=-\tr(\rho_V\log(\rho_V))$ denotes the von Neumann entropy of reduced density $\rho_V$ (and similarly defined for other terms). We say $\mc N$ is \textit{less noisy}, if for any classical-quantum state $\rho_{\mc X A} = \sum_{x \in \mc X} p_x \ketbra{x}{x} \otimes \rho_A^x$, we have 
\begin{align*}
    I(\mc X;B)_{(\id_{\mc X}\otimes \mc N) (\rho_{VA})} \ge I(\mc X;E)_{(\id_{\mc X}\otimes \mc N) (\rho_{VA})}. 
\end{align*}
In the following, we will often write $I(V;B)$
when the underlying state is clear from the context.

Note that there exist two different notions of less noisy quantum channels in the literature. One, which we are exclusively using in this work and in the definition above, refers to the classical-quantum mutual information with respect to a single application of the channel $\mc N$ and $\mc N^c$ \cite{TBR20,HRS22} (which can also be called \textit{level-1 less noisy}). This notion characterizes the class of quantum channels with concave coherent information. In fact, for a quantum channel $\mc N$ and an input state $\rho_A$ with purification $\ket{\psi}_{A'A}$, the coherent information is defined as the coherent information of the bipartite state $\rho_{A'B} = (\id_{A'} \otimes \mc N) (\ketbra{\psi}{\psi}_{A'A})$:
\begin{equation}\label{def:coherent information}
    I_c(\mc N,\rho_A):= I(A'\rangle B)_{\rho_{A'B}} = S(B) - S(A'B) = S(B) - S(E).
\end{equation}
Then, concavity of this quantity means for any ensemble of states $\{p_x,\rho_A^x\}_{x \in \mc X}$, we have 
\begin{align*}
    I_c(\mc N,\sum_xp_x\rho_A^x) \ge \sum_x p_x I_c(\mc N,\rho_A^x),
\end{align*}
which is equivalent to $I(\mc X;B) \ge I(\mc X;E)$. 

Another notion refers to a regularized version for many copies of $\mc N$ introduced in \cite{Watanabe12} (sometimes called \textit{regularized less noisy}), which implies that the private information and coherent information are weakly additive for this channel.

It is clear that informational degradablity implies less noisy by restricting the general bipartite density operators $\rho_{VA}$ to be a classical-quantum state. Moreover, via data processing inequality, degradablity implies informational degradablility. By this reasoning, any channel that is degradable is also less noisy. However, to the best of our knowledge, it was an open question whether there exists a level-1 less noisy quantum channel that is not degradable\footnote{It remains an open question whether there exists a regularized less noisy channel which is not degradable.}, which we resolve in this work. To this end, we propose a framework in Section \ref{sec: applications} for constructing such examples and give an explicit example in terms of amplitude damping channels. This framework may further be used to construct examples of non-degradable channels that are informationally degradable. 

\section{Impossibility of a reverse data processing inequality for non-unitary channels}
\label{sec: impossibility of reverse data processing}
In this section, we show that a reverse data processing inequality cannot hold for non-unitary channels $\mc N:\mb B(\mc H_A)\to \mb B(\mc H_B)$ with dimensions $d_A\geq d_B$. More precisely, we show that expansion coefficient $\widecheck{\eta}_{\mc N}$ in this setting generically equals to zero.
\medskip
\begin{theorem}\label{main:impossibility}
    Let $\mc N\in \mc L(\mb B(\mc H_A), \mb B(\mc H_B))$ be a quantum channel such that $d_A \ge d_B$. Then,  %$\text{dim} \mc H_A \ge \text{dim} \mc H_B$, then 
    \begin{equation}
    \widecheck{\eta}_{\mc N}=% \inf_{\rho \neq \sigma} \frac{D(\mc N(\rho) || \mc N(\sigma))}{D(\rho||\sigma)} = 
    \begin{cases}
        1,\quad d_A =d_B \text{ and } \ \mc N(\rho) = U\rho U^{\dagger} \text{ for some unitary U}, \\
        0,\quad \text{otherwise}.
    \end{cases}
    \end{equation}
\end{theorem}

%In other words, for these channels, it is always possible to find two (non-equal) quantum states that remain indistinguishable after the channel is applied.
The same conclusion does not holds for channels with strictly greater output dimension than input dimension.  For example, the erasure channel with erasure probability $\nu \in [0,1)$ 
\[\mc N(\rho)= (1-\nu)\rho+ \nu \ketbra{e}{e}\]
is a simple counterexample with $\widecheck{\eta}_{\mc N}=1-\nu>0$. 

The key ingredient to prove the above theorem is the following lemma about purity-preserving quantum channels. 
The proof can be found in \cite[Theorem 3.1]{davies1976quantum}. For the convenience of the reader, we present an independent proof below.
\medskip
\begin{lemma}\label{pure state argument}
   If a quantum channel $\mc N\in \mc L(\mb B(\mc H_A), \mb B(\mc H_B))$ preserves the purity, i.e., it maps any pure state to a pure state, then $\mc N$ must either be an isometric embedding $\mc N(\rho) = V\rho V^{\dagger}, \ V^{\dagger} V = I_A$ or a replacer channel $\mc N(\rho) = \tr(\rho) \ketbra{\varphi}{\varphi}$ for some pure state $\ket{\varphi}$.
\end{lemma}

\begin{proof}[Proof of Lemma~\ref{pure state argument}]
Let $\{E_i\}_{i=1}^k$ denote the set of linearly independent Kraus operators with $\sum_{i=1}^k E_i^{\dagger} E_i = I$ that form the minimal Kraus representation of the channel $\mc N\in \mc L(\mb B(\mc H_A), \mb B(\mc H_B))$, i.e. \[\mc N(\rho) = \sum_{i=1}^k E_i \rho E_i^{\dagger}\ \forall\rho\in B(\mc H_A)\ .\]
If $k=1$, $E_1^\dagger E_1=I$ is an isometry which can only happen if $\text{dim} \mc H_A \le \text{dim} \mc H_B$.
We now show that, if $k>1$ and $\mc N$ is a quantum channel that preserves purity, $\mc N$ must be a replacer channel. 

%\paula{for all i? for one i? in order for it to preserve purity?} 
Suppose that $k>1$ and that $\mc N$ is a purity-preserving quantum channel. In this case, we claim that $\forall 1\le i \le k$, the codimension of $\text{Ker} E_i = \{x \in \mc H_A: E_i x = 0\}$ must be 1, i.e., 
\begin{align*}
    H_i := \text{Ker} E_i,\quad \text{dim} H_i^{\perp} = 1.
\end{align*}
We only show the case $i=1$ because the same argument applies to other $E_i$. We argue by contradiction. Suppose $\text{dim} H_1^{\perp} >1$. Then one can show there exists $0\neq \mu_0 \in \mb C$ such that 
\begin{align}\label{lemma:uniform constant orthogonal}
    E_1\big|_{H_1^{\perp}} = \mu_0 E_2\big|_{H_1^{\perp}}.
\end{align}
In fact, for any orthogonal pure states $\ket{\varphi_1}, \ket{\varphi_2}$ in $H_1^{\perp}$, because $\mc N$ must map $\ket{\varphi_1}$, $\ket{\varphi_2}$ as well as their linear combination to pure states, there exist complex constants $c_1,c_2,c_3$ such that 
\begin{align*}
    & E_2 \ket{\varphi_1} = c_1  E_1 \ket{\varphi_1}, \\
    & E_2 \ket{\varphi_2} = c_2  E_1 \ket{\varphi_2}, \\
    & E_2 (\ket{\varphi_1}+\ket{\varphi_2}) = c_3  E_1 (\ket{\varphi_1}+\ket{\varphi_2}). 
\end{align*}
Since $E_1 \ket{\varphi_1}$ and $E_1 \ket{\varphi_2}$ are linearly independent, we must have $c_1 = c_2 = c_3$. As this holds for arbitrary two vectors $\ket{\varphi_1}, \ket{\varphi_2}$ in $H_1^{\perp}$, \eqref{lemma:uniform constant orthogonal} holds. Next, we show that we also have \begin{align}\label{lemma:uniform constant kernel}
    E_1\big|_{H_1} = \mu_0 E_2\big|_{H_1} = 0.
\end{align}
In fact, for any $\ket{\psi} \in H_1$ and any orthogonal pure states $\ket{\varphi_1}, \ket{\varphi_2}$ in $H_1^{\perp}$, there exist non-zero complex constants $c_1', c_2'$ such that
\begin{align*}
    & E_1\ket{\varphi_1} = E_1 (\ket{\psi}+\ket{\varphi_1}) = c_1'  E_2(\ket{\psi}+\ket{\varphi_1}) = c_1'  E_2\ket{\psi}+ c_1'\mu_0  E_1 \ket{\varphi_1}, \\
    & E_1\ket{\varphi_2} = E_1 (\ket{\psi}+\ket{\varphi_2}) = c_2'  E_2(\ket{\psi}+\ket{\varphi_2}) = c_2'  E_2\ket{\psi}+ c_1'\mu_0  E_1 \ket{\varphi_2}, 
\end{align*}
which shows that $E_2\ket{\psi}$ is parallel to $E_1\ket{\varphi_1}$ and $E_1\ket{\varphi_2}$ simultaneously thus $E_2\ket{\psi} = 0$. Therefore, \eqref{lemma:uniform constant orthogonal} and \eqref{lemma:uniform constant kernel} hold thus $E_1 = \mu_0 E_2$ which contradicts the fact that $E_1$ and $E_2$ are linearly independent. Therefore, for every $1\le i \le k$,  we must have $\text{dim} H_i^{\perp} = 1$, hence $E_i$ is rank 1 operator
\begin{align*}
    E_i = \ket{\varphi} \bra{\psi_j}
\end{align*}
and $\mc N(\rho) = \tr(\rho)\ketbra{\varphi}{\varphi}$ which concludes the proof.
\end{proof}

\noindent We are now in a position to prove Theorem~\ref{main:impossibility}.

\begin{proof}[Proof of Theorem~\ref{main:impossibility}]
If $\mc N(\rho)=\tr(\rho) \ketbra{\varphi}{\varphi}$ is a replacer channel, $\widecheck{\eta}_{\mc N}=0$ because the numerator $D(\mc N(\rho)||\mc N(\sigma))=0$ is always zero.
For unitary channel $\mc N(\rho) = U\rho U^{\dagger}$, the expansion coefficient $\widecheck{\eta}_{\mc N}$  equals to $1$ due to the unitary invariance of the relative entropy. By the above Lemma~\ref{pure state argument}, it suffices to consider channels that are not purity-preserving. 

In this case, we claim that one can find a projection $P_A$ onto a subspace of $\mc H_A$ with dimension $\dim(P_A(\mc H_A))\le d_A - 1$ and a pure state $\ket{\psi}\in P_A(\mc H_A)^\perp$ such that 
\begin{align*}
        \text{supp} \left(\mc N(P_A)\right) = \text{supp} \left(\mc N(P_A + \ketbra{\psi}{\psi}) \right).
\end{align*}
In order to construct $P_A$, we begin by choosing a pure state $\ket{\varphi_1}$ such that $\mc N(\ketbra{\varphi_1}{\varphi_1})$ is a mixed state, which is possible as $\mc N$ is not purity-preserving. Then, we extend %denote $d = \text{dim} \mc H_A$ and extend
$\ket{\varphi_1}$ to an orthonormal basis $\{\ket{\varphi_i}\}_{1 \le i \le d_A}$ and get a family of projections 
$$P_k = \sum_{i=1}^k \ketbra{\varphi_i}{\varphi_i},\quad 1\le k \le d_A.$$
The support of $\mc N(P_k)$ is a chain of subspaces of $\mc H_B$ that fulfills
\begin{align*}
        \text{supp} \left(\mc N(P_1)\right) \subseteq \text{supp} \left(\mc N(P_2)\right) \subseteq \cdots \subseteq \text{supp} \left(\mc N(P_k)\right) \subseteq \cdots  \subseteq \text{supp} \left(\mc N(P_{d})\right).
\end{align*}
Recall that, by assumption, $\mc N(P_1) = \mc N(\ketbra{\varphi_1}{\varphi_1})$ is a mixed state, and thus the dimension of $\dim \text{supp} \left(\mc N(P_1)\right)\ge 2$. Thus 
\begin{align*}
2 \le \text{dim}\left(\text{supp} \left(\mc N(P_1)\right)\right)  \le \cdots \le \text{dim}\left(\text{supp} \left(\mc N(P_k)\right)\right) \le \cdots \text{dim}\left(\text{supp} \left(\mc N(P_{d})\right)\right) \le \text{dim} \mc H_B \le d_A.
\end{align*}
Since there are $d_A$ many subspaces and the dimension can take at most $d_A-1$ values, there exists $k_0<d_A$ such that 
\begin{align*}
    \text{dim}\left(\text{supp} \left(\mc N(P_{k_0})\right)\right) = \text{dim}\left(\text{supp} \left(\mc N(P_{k_0 + 1})\right)\right)
\end{align*}
thus $\text{supp} \left(\mc N(P_{k_0})\right) = \text{supp} \left(\mc N(P_{k_0+1})\right)$. Then, the claim is verified by choosing
\begin{align*}
    P_A = P_{k_0}, \quad \ket{\psi} = \ket{\varphi_{k_0+1}}.
\end{align*}
Now, we use this construction to show $\widecheck{\eta}_{\mc N}=0$. Denote %Denote 
    \begin{align*}
        \rho = \frac{1}{k_0}P_A,\quad \sigma_{\varepsilon} = (1-\varepsilon)\rho + \varepsilon \ketbra{\psi}{\psi}.
    \end{align*}
    By direct calculation, 
    \begin{align*}
        & D(\rho \|\sigma_{\varepsilon}) = -\log(1-\varepsilon), \\
        & \frac{\dd}{\dd\varepsilon}D(\mc N(\rho) \|\mc N(\sigma_{\varepsilon}))\big|_{\varepsilon = 0} = 0.
    \end{align*}
   The second equation follows from the fact that $\text{supp} \left(\mc N(\ketbra{\psi}{\psi})\right) \subseteq \text{supp} \left(\mc N(P_A)\right)$, thus there exists $\varepsilon_0>0$ such that for any $\varepsilon \in (-\varepsilon_0,\varepsilon_0)$, $\mc N(\sigma_{\varepsilon})$ is a density operator. Therefore, the non-negative, differentiable function defined as $f(\varepsilon) = D(\mc N(\rho) \|\mc N(\sigma_{\varepsilon})) \ge 0$ achieves its minimum at $\varepsilon = 0$. Therefore, the derivative at $\varepsilon = 0$ is zero.
    
    Inserting the states $\rho$ and $\sigma_{\varepsilon}$ in the expansion coefficient and letting $\varepsilon$ go to zero, we have by L'H\^opital's rule \begin{align*}
        0& \le \inf_{\rho \neq \sigma} \frac{D(\mc N(\rho) || \mc N(\sigma))}{D(\rho || \sigma)} \le \lim_{\varepsilon \to 0} \frac{D(\mc N(\rho) \|\mc N(\sigma_{\varepsilon}))}{D(\rho \|\sigma_{\varepsilon})} \\& = \lim_{\varepsilon \to 0}\frac{\frac{\dd}{\dd\varepsilon} D(\mc N(\rho) \|  \mc N(\sigma_{\varepsilon}))}{\frac{\dd}{\dd\varepsilon} D(\rho \|\sigma_{\varepsilon})}  = \lim_{\varepsilon \to 0} (1-\varepsilon)\frac{\dd}{\dd\varepsilon} D(\mc N(\rho) \|  \mc N(\sigma_{\varepsilon}))  = 0.
    \end{align*}
\end{proof}
\begin{remark}{\rm
\added{In \cite{RISB21}, a positive expansion coefficient is proposed as a condition under which the proposed quantum version of the Blahut-Arimoto algorithm for computing quantum channel capacities converges exponentially fast. More precisely, they show exponential convergence for the computation of the Holevo quantity under the assumption in \cite[Eq.~31]{RISB21}, which is equivalent (by \cite[Eq.~45]{RISB21}) to a positive expansion coefficient for the channel. Since we show here that this assumption cannot hold for non-unitary channels with $d_A\geq d_B$, the proposed Blahut-Arimoto algorithm for the Holevo quantity for these channels thus does not fulfill the criteria for exponential convergence and is only proven to have polynomial convergence.}
}
\end{remark}
\medskip
\added{
\begin{remark}{\rm
Note that our result does not imply that $\inf_{\rho} \frac{D(\mc N(\rho) || \mc N(\sigma))}{D(\rho || \sigma)} = 0$ for any fixed state $\sigma$.
For example, fixing $\sigma=\frac{I_2}{2}$, for the qubit depolarizing channel $\mc D_p (X):= (1-p)X + \frac{p}{2} I_2$, using the calculation in Proposition \ref{prop: qubit depolarizing}, we can show that 
\begin{equation}
    \inf_{\rho} \frac{D(\mc D_p(\rho) || \mc D_p(\frac{I_2}{2}))}{D(\rho || \frac{I_2}{2})} > 0.
\end{equation}
}
\end{remark}
}

\section{Relative contraction and expansion for pairs of quantum channels}\label{sec: relative contraction and expansion}
\added{Motivated by the vanishing of expansion coefficient of a single channel $\mc N$ shown in Section \ref{sec: impossibility of reverse data processing}, in this section, we present three different methods for comparing the expansion and the contraction of the relative entropy for two channels $\mc N$ and $\mc M$.} Recall that
\begin{align}
 \widecheck{\eta}_{\mc N,\mc M} := \inf_{\rho \neq \sigma, \supp(\rho) \subseteq \supp(\sigma)} \frac{D(\mc N(\rho) \| \mc N(\sigma))}{D(\mc M(\rho) \| \mc M(\sigma))}\ , \  \ \   \eta_{\mc N,\mc M} := \sup_{\rho \neq \sigma,\supp(\rho) \subseteq \supp(\sigma)} \frac{D(\mc N(\rho) \| \mc N(\sigma))}{D(\mc M(\rho) \| \mc M(\sigma))}. \label{def:rel-expan1}
\end{align}
We introduce several techniques to bound or compute the relative coefficients in this context. The first technique leverages the equivalence between the relative expansion of the relative entropy and its infinitesimal counterpart, the Bogoliubov-Kubo-Mori (BKM) metric (see Lemma~\ref{lemma: comparison}). The second technique employs a completely positive (CP) order comparison of two channels, as established in Lemma~\ref{lemma: comparison dephasing}. Additionally, we conduct a detailed study of qubit channels using the Bloch vector representation.

These techniques are applied in Section~\ref{sec: examples}, where we provide examples of channel pairs with non-zero relative expansion coefficients.
% is based on the equivalence of relative expansion or contraction of the relative entropy, and relative expansion or contraction of a particular norm, which we prove in Lemma~\ref{lemma: comparison}.

\subsection{Comparison of BKM metric}
\label{sec:BKM}
Our starting point is the following integral representation of the relative entropy $D(\rho\|\sigma)$ from \cite[Lemma 2.2]{GR22}, which is also studied in \cite{LR99, AE11,HR15,GR22,WBCDT24}:
\begin{align}\label{eqn:integral representation L2}
    D(\rho\|\sigma) =  \int_0^1 \int_0^s g_{\rho_t}(\rho-\sigma) \dd t \dd s . 
\end{align}
where $\rho_t:=(1-t) \sigma+t \rho, t \in[0,1]$ and the BKM metric $g_\sigma(X)$ of an operator $X$ at density $\sigma$ is defined as
\[g_\sigma(X)= \begin{cases}
        \displaystyle \tr\left(\int_0^{\infty}X^\dagger (\sigma + rI)^{-1} X (\sigma +rI)^{-1} \dd r\right), \quad &\text{supp}(X) \subseteq \text{supp}(\sigma) \\
        \infty,\quad  &\text{else}.
    \end{cases}\]
In fact, define a function $f(t)= D\left(\rho_t \| \sigma\right), t \in[0,1]$. We have $f(0)=0, f(1)=D(\rho \| \sigma)$ and the derivatives
\begin{align}
f^{\prime}(t) & =\operatorname{tr}\left((\rho-\sigma) \ln \rho_t-(\rho-\sigma) \ln \sigma\right) 
\\ f^{\prime \prime}(t) & =\int_0^{\infty} \operatorname{tr}\left((\rho-\sigma) (\rho_t + rI)^{-1}(\rho-\sigma) (\rho_t + rI)^{-1}\right) \dd r=g_{\rho_t}(\rho-\sigma).\label{eq:2nd}
\end{align}

Since $f^{\prime}(0)=0$, and the integral representation \eqref{eqn:integral representation L2} follows from  \[\displaystyle D(\rho \| \sigma) =f(1)=\int_0^1\left(\int_0^s f^{\prime \prime}(t) \dd t\right) \dd t\dd s\ .\] 

Given any density operator $\sigma$ acting on $\mc H$, we also define the BKM operator  
\begin{align}\label{eq:BKMop} \mc J_{\sigma} (X) = \int_0^{\infty} (\sigma + rI)^{-1} X (\sigma + rI)^{-1} \dd r, \quad &\text{supp}(X) \subseteq \text{supp}(\sigma).\end{align}
It is clear that $g_\sigma(X)=\langle X, \mc J_{\sigma} (X)\rangle$ with respect to the trace inner product $\langle Y,X\rangle=\tr(Y^\dagger X)$.
%and we have $D(\rho \| \sigma) =\int_0^1 \int_0^s J\dd t \dd s$. 
The following lemma gives a criterion for the comparison between the relative entropies $D(\mc M(\rho)\|\mc M(\sigma))$ and $D(\mc N(\rho)\|\mc N(\sigma))$ via comparison of BKM metric $g_{\mc N(\sigma)}(\mc N(X)) $ and $g_{\mc M(\sigma)}(\mc M(X)) $. The equivalence between contraction coefficient of relative entropy and BKM metric when one of the channel is identity was previously studied in \cite{LR99}. In the following lemma, we answer the open question following Theorem 7.1 in \cite{HR15}, and extend it to the case where two arbitrary channels are compared.
\medskip
\begin{lemma}\label{lemma: criteria general}
    Let $\mc N\in \mc L(\mb B(\mc H_A), \mb B(\mc H_B))$ and $\mc M\in \mc L(\mb B(\mc H_A), \mb B(\mc H_B'))$ be two quantum channels. For any $c_1,c_2 >0$, the following two statement are equivalent:  
    \begin{enumerate}[(i)]
        \item \label{lemma:criteria1} For any density operators $\rho$ and $\sigma$, with $\supp(\rho) \subseteq \supp(\sigma)$,
        \begin{align*}
            c_1 D(\mc N(\rho)\|\mc N(\sigma)) \leq D(\mc M(\rho)\|\mc M(\sigma)) \leq c_2 D(\mc N(\rho)\|\mc N(\sigma)). 
        \end{align*}
        \item \label{lemma:criteria2}For any density operators $\sigma$ and traceless Hermitian operator $X$ with $\supp(X) \subseteq \supp(\sigma)$, 
    \begin{equation}\label{eqn: criteria general}
    \begin{aligned}
        c_1  g_{\mc N(\sigma)}(\mc N(X)) & \leq g_{\mc M(\sigma)}(\mc M(X))\leq c_2 g_{\mc N(\sigma)}(\mc N(X)).
    \end{aligned}
\end{equation}
    \end{enumerate}
\end{lemma}
\begin{proof}Take $\rho_t = (1-t)\sigma + t \rho=\sigma +tX$ and $X = \rho - \sigma$. For any $t\in (0,1)$, $\text{supp}(X)\subseteq \text{supp}(\rho_t)$.  The direction \ref{lemma:criteria2}$\implies$\ref{lemma:criteria1}  follows from the integral representation \eqref{eqn:integral representation L2}.

To prove \ref{lemma:criteria1}$\implies$\ref{lemma:criteria2}, switching the roles of $\mc N$ and $\mc M$, we only need to show that if $c_1 D(\mc N(\rho)\|\mc N(\sigma)) \leq D(\mc M(\rho)\|\mc M(\sigma))$ for any density operators $\rho,\sigma$ with $\supp(\rho) \subseteq \supp(\sigma)$, then we have $c_1  g_{\mc N(\sigma)}(\mc N(X)) \leq g_{\mc M(\sigma)}(\mc M(X))$ for any density operators $\sigma$ and traceless Hermitian operator $X$ with $\supp(X) \subseteq \supp(\sigma)$. In fact, define $\rho_t = \sigma + tX$, $\mc N(\rho_t)$ and $\mc M(\rho_t)$ are density operators for $t \in (-\varepsilon,\varepsilon)$ with $\varepsilon>0$ sufficiently small, and we have 
    \begin{align*}
        c_1 D(\mc N(\rho_t)\|\mc N(\sigma)) \leq D(\mc M(\rho_t)\|\mc M(\sigma)).
    \end{align*}
    Note that 
    \begin{align*}
        & D(\mc N(\rho_t)\|\mc N(\sigma))\big|_{t=0} = D(\mc M(\rho_t)\|\mc M(\sigma))\big|_{t=0} = 0, \\
        & \frac{\dd}{\dd t}D(\mc N(\rho_t)\|\mc N(\sigma))\big|_{t=0} = \frac{\dd}{\dd t}D(\mc M(\rho_t)\|\mc M(\sigma))\big|_{t=0} = 0.
    \end{align*}
    Thus we have the second order comparison: 
    \begin{align*}
        c_1 \frac{\dd^2}{\dd t^2} D(\mc N(\rho_t)\|\mc N(\sigma))\big|_{t=0} \leq \frac{\dd^2}{\dd t^2} D(\mc M(\rho_t)\|\mc M(\sigma))\big|_{t=0},
    \end{align*}
    which concludes the proof of $c_1  g_{\mc N(\sigma)}(\mc N(X)) \leq g_{\mc M(\sigma)}(\mc M(X))$ by expanding the second-order derivative \eqref{eq:2nd}.  
\end{proof}

\noindent We will also make use of the following result from \cite[Lemma 2.1]{GR22}:
\medskip
\begin{lemma}[{\cite[Lemma 2.1]{GR22}}]
\label{lemma: comparison}
    If two density operators $\rho$ and $\sigma$  satisfy $\rho \le c \sigma$ for some $c>0$, then for any operator $X \in \mb B(\mc H)$, 
    \begin{equation}
       g_{\rho}(X)\ge \frac{1}{c}  g_{\sigma}(X).
    \end{equation}
\end{lemma}
\noindent For the convenience of the reader, the proof is provided in Appendix~\ref{appendix: proof of comparison}.

\subsection{Comparison of completely positive order} 
\label{sec:CP}
Here, we propose a criterion based on CP order which will later be used for computing the relative expansion coefficient of two dephasing channels in Section~\ref{sec: example dephasing}. Suppose that two channels $\mc N$ and $\mc M$ satisfy
\begin{equation}\label{eqn: cp comparison}
   c_1 \mc N \le_{cp} \mc M \le_{cp} c_2 \mc N 
\end{equation}
for some positive constants $c_1,c_2 >0$. \added{We note that this condition is equivalent to the comparison of Choi matrices of channels in terms of positive semidefiniteness
\[ c_1 C_{\mc N} \le_{cp} C_{\mc M} \le_{cp} c_2 C_{\mc N} \]
where $c_1 \mc C_{\mc N}\le \mc C_{\mc M}$ means $\mc C_{\mc M} - c_1 \mc C_{\mc N}$ is positive semidefinite.}
By Lemma \ref{lemma: comparison}, for any operator $Y$ and density $\omega$,
\[  \frac{1}{c_2}  g_{\mc M(\omega)}(Y)\le g_{\mc N(\omega)}(Y)\le \frac{1}{c_1}  g_{\mc M(\omega)}(Y).\]
Then the target inequality % using Lemma~\ref{lemma: comparison}, the target inequality from \ref{lemma: criteria general}.
\begin{align*}
    g_{ \mc N(\omega)}(\mc N(X))\ge cg_{ \mc M(\omega)}(\mc M(X)),\quad  \forall \ \omega,X
\end{align*}
can be deduced from
\begin{align}
g_{\omega}(\mc N(X)) \geq c' g_{\omega}(\mc M(X))\label{eqn:simplified target inequality}
\end{align}
for some $c'>0$ (Either $c'=cc_1$ and $\omega=\mathcal{N}(\rho_t),t\in [0,1]$, or $c'=c c_2$ for $\omega=\mathcal{M}(\rho_t), t\in [0,1]$).

%with $c'=cc_1$ for $\omega=\mathcal{N}(\rho_t)$, or $c'=c c_2$ for $\omega=\mathcal{M}(\rho_t)$, in order to have .
%where $H$ is given by the difference of density operators and $\omega$ is a density operator given by either $\mc M(\rho_t)$ or $\mc N(\rho_t)$. 

% %where $c_1 \mc C_{\mc N}\le \mc C_{\mc M}$ means $\mc C_{\mc M} - c_1 \mc C_{\mc N}$ is positive semidefinite.
% Under this assumption \eqref{eqn: cp comparison}, the target inequality % using Lemma~\ref{lemma: comparison}, the target inequality from \ref{lemma: criteria general}.
% \begin{align*}
%     \langle \mc M(\rho-\sigma), \mc J_{\mc M(\rho_t)} (\mc M(\rho-\sigma))\rangle \ge c \langle \mc N(\rho-\sigma), \mc J_{\mc N(\rho_t)} (\mc N(\rho-\sigma))\rangle, \forall \rho,\sigma,t\in [0,1]
% \end{align*}
% simplifies to \begin{align}\label{eqn:simplified target inequality}
%        \langle \mc M(X), \mc J_{\omega} (\mc M(X))\rangle \geq c' \langle \mc N(X), \mc J_{\omega} (\mc N(X))\rangle
% \end{align}
% where $H$ is given by the difference of density operators and $\omega$ is a density operator given by either $\mc M(\rho_t)$ or $\mc N(\rho_t)$. 
It is tempting to conjecture that the comparison in CP order \eqref{eqn: cp comparison} \added{(or equivalently, positive semidefinite order in terms of Choi matrices)} directly implies the comparison of BKM metric \eqref{eqn:simplified target inequality}. The latter is equivalent to 
\begin{align}\label{ineqn:L^2 positiveness}
     \langle X, \mc N^{\dagger} \mc J_{\omega} \mc N(X) \rangle \ge c'\langle X, \mc M^{\dagger} \mc J_{\omega} \mc M (X) \rangle, \forall X
\end{align}
where $\langle X,Y\rangle = \tr(X^{\dagger}Y)$ is the standard Hilbert-Schmidt inner product. 
However, the complete positivity of superoperators does not imply positive semidefiniteness as an operator on the Hilbert-Schmidt space. In fact, suppose we have a completely positive map $\Psi(\rho) = \sigma_z\rho \sigma_z$, one can easily show that $\langle X,\Psi(X)\rangle < 0$ for a Hermitian operator $X = \begin{pmatrix}
    a & z \\
    z^* & -a
\end{pmatrix}$ with $|a| < |z|$. Hence, $\Psi$ is completely positive map but is not a positive semidefinite as an operator on the Hilbert-Schmidt space. It is therefore not enough to assume \eqref{eqn: cp comparison} in order to have a nontrivial expansion coefficient through Lemma~\ref{lemma: criteria general}. 

Instead, we need an additional assumption in order to guarantee that the comparison from \eqref{eqn:simplified target inequality} holds for some $c'>0$. The following lemma is motivated by \cite[Lemma 2.3]{GJLL22}.
\medskip
\begin{lemma}\label{lemma: comparison dephasing}
 Suppose $\mc M,\mc N,\Phi \in \mc L(\mb B(\mc H_A), \mb B(\mc H_B))$ are quantum channels such that 
    \begin{align*}
        \mc N = (1-\varepsilon)\mc M + \varepsilon \Phi, \quad \varepsilon \in (0,1).
    \end{align*}
    Moreover, we assume that there exists a quantum channel $\mc D \in \mc L(\mb B(\mc H_A), \mb B(\mc H_A))$ such that $\mc D \circ \mc N = \Phi$ and $\mc D(\omega) \le c \omega$ for some fixed density operator $\omega$ and $c>0$. Then, for any operator $X$, we have 
    \begin{equation}
        g_\omega(\mc N(X))\ge \frac{(1-2\varepsilon)(1-\varepsilon)}{1 + c \varepsilon(1-\varepsilon)}g_\omega(\mc M(X)).
    \end{equation}
\end{lemma}
\begin{proof}
   Recall that for the BKM metric $g_{\omega}: \mb B(\mc H) \to \mb [0,\infty]$,
$g_{\omega}(X) = \langle X, \mc J_{\omega} (X)\rangle$, $X\mapsto \sqrt{g_{\omega}(X)}$ is a Hilbert space norm on $\text{supp}(\omega)$. By triangle inequality, we have 
    \begin{align*}
       \sqrt{g_{\omega}(\mc N(X))} &= \sqrt{g_{\omega}((1-\varepsilon)\mc M(X) + \varepsilon \Phi (X))} \\
       & \ge \sqrt{g_{\omega}((1-\varepsilon) \mc M(X))} - \sqrt{g_{\omega}(\varepsilon \Phi (X))} \\
       & = (1-\varepsilon) \sqrt{g_{\omega}(\mc M(X))} - \varepsilon \sqrt{g_{\omega}(\Phi (X))}.
    \end{align*}
    Taking the square on both sides, 
    \begin{align*}
        g_{\omega}(\mc N(X)) & \ge (1-\varepsilon)^2 g_{\omega}(\mc M(X)) - 2\varepsilon(1-\varepsilon) \sqrt{g_{\omega}(\mc M(X))} \sqrt{g_{\omega}(\Phi (X))} + \varepsilon^2 g_{\omega}(\Phi(X)) \\
        & \ge (1-\varepsilon)^2 g_{\omega}(\mc M(X)) - 2\varepsilon(1-\varepsilon) \sqrt{g_{\omega}(\mc M(X))} \sqrt{g_{\omega}(\Phi (X))} \\
        & \ge (1-\varepsilon)^2 g_{\omega}(\mc M(X)) - \varepsilon(1-\varepsilon) \big( g_{\omega}(\mc M(X)) + g_{\omega}(\Phi (X)) \big).
    \end{align*}
    To compare $g_{\omega}(\Phi (X))$ and $g_{\omega}(\mc N(X))$, we use Lemma~\ref{lemma: comparison} to get 
    \begin{align*}
       g_{\omega}(\Phi (X)) \le cg_{\mc D(\omega)}(\Phi (X)) = cg_{\mc D(\omega)}(\mc D \circ \mc N (X)) \le c g_{\omega}(\mc N(X)),
    \end{align*}
    where the last inequality is the data processing inequality of the BKM metric \cite{LR99}. Finally, using the above inequalities, we get
    \begin{align*}
        g_{\omega}(\mc N(X)) & \ge (1-\varepsilon)^2 g_{\omega}(\mc M(X)) - \varepsilon(1-\varepsilon) \big( g_{\omega}(\mc M(X)) + g_{\omega}(\Phi (X)) \big) \\
        & \ge (1-2\varepsilon)(1-\varepsilon)g_{\omega}(\mc M(X)) - c\varepsilon(1-\varepsilon)g_{\omega}(\mc N(X)),
    \end{align*}
    which implies 
    \begin{align*}
        g_{\omega}(\mc N(X)) \ge \frac{(1-2\varepsilon)(1-\varepsilon)}{1 + c \varepsilon(1-\varepsilon)} g_{\omega}(\mc M(X)).
    \end{align*}
\end{proof}

\begin{comment}
We will use two expressions for qubit density operators. The first is given in terms of the state's eigenvalues in Lemma~\ref{lemma: qubit BKM metric} and the second is given in terms of the state's Bloch vector in Lemma~\ref{lemma: qubit BKM metric}.
%Before we proceed to prove the result, recall that an explicit calculation of BKM metric $\langle X, \mc J_{\sigma}(X)\rangle$ for any qubit density operator $\sigma$ and any operator $X$ was known, see for example \cite[Eq.~16]{AE11}. 
\begin{lemma}\label{lemma: qubit BKM metric}
Let a qubit density operator $\sigma$ and a general matrix $X$ be given as 
\begin{align*}
\sigma = U \begin{pmatrix} 1-\lambda & 0 \\0 & \lambda \end{pmatrix} U^{\dagger}, \quad X = \begin{pmatrix} a_{11} & a_{12} \\ a_{21} & a_{22} \end{pmatrix},\quad \lambda\in (0,1). 
\end{align*}
Then we have 
\begin{align*}
   \langle X, \mc J_{\sigma}(X) \rangle & = \int_0^{\infty} \tr\left(X^{\dagger} (\sigma + sI)^{-1} X (\sigma + sI)^{-1} \right) ds \\
   & = \frac{1}{1 - \lambda} |\wt a_{11}|^2 + \frac{1}{\lambda} |\wt a_{22}|^2 + \frac{\log (\frac{1-\lambda}{\lambda}) }{1 - 2\lambda}\left(|\wt a_{12}|^2 + |\wt a_{21}|^2\right),
\end{align*}
where $U^{\dagger} X U = \begin{pmatrix}
\wt a_{11} & \wt a_{12} \\ 
\wt a_{21} & \wt a_{22}
\end{pmatrix} $. 
\end{lemma}
This was previously shown in, for example, \cite[Eq.~16]{AE11}. 
\end{comment}
\subsection{Explicit formula for the qubit channels}\label{sec:qubit calculation}
In this section, we discuss the qubit case and provide a sufficient condition for $\widecheck{\eta}_{\mc N,\mc M} >0$ which can cover a large family of examples. Recall that the identity and Pauli matrices 
\begin{equation}
    \sigma_x = \begin{pmatrix}
        0 & 1 \\
        1 & 0
    \end{pmatrix},\quad \sigma_y = \begin{pmatrix}
        0 & i \\
        -i & 0
    \end{pmatrix},\quad \sigma_z = \begin{pmatrix}
        1 & 0 \\
        0 & -1
    \end{pmatrix}.
\end{equation}
together form an orthonormal basis for $\mb M_2$. 
Any traceless Hermitian operator $X$ and density operator $\rho$ can be represented by two real vectors: 
\begin{equation}\label{eqn: qubit representation}
\begin{aligned}
     & X = \vec{y} \cdot \vec{\sigma} = y_1 \sigma_x + y_2 \sigma_y + y_3 \sigma_z, \quad \vec{y} \in \mb R^3   \\
     & \rho = \frac{1}{2} (\mb I_2 + \vec{w}\cdot \vec{\sigma}) = \frac{1}{2} (\mb I_2 + w_1 \sigma_x + w_2 \sigma_y + w_3 \sigma_z),\quad \vec{w}\in \mb R^3.
\end{aligned}
\end{equation}
where $\vec{\sigma}= (\sigma_x,  \sigma_y, \sigma_z)$ denotes the vector of Pauli matrix. 
Note that $\rho$ is a density operator if and only if $|\vec w| \le 1$. Thus the set of density operators can be identified with the unit ball in $\mb R^3$ and the pure states lie on the Bloch sphere. The Pauli basis has also been used to study the contraction coefficient of unital qubit channel by Hiai and Ruskai \cite{HR15}. The following basic properties are useful, see \cite[Appendix B]{HR15}:
\begin{equation}\label{eqn: qubit representation calculation}
\begin{aligned}
     & \text{Product rule:}\ (a \mb I_2 + \vec{w} \cdot \vec{\sigma})(b \mb I_2 + \vec{y}\cdot \vec{\sigma}) = (ab + \vec{w} \cdot \vec{y}) \mb I_2 + (a\vec{y} + b \vec{w} + i \vec{w} \times \vec{y})\cdot \vec{\sigma},  \\
     & \text{Inverse rule:}\ (a \mb I_2 + \vec{w} \cdot \vec{\sigma})^{-1} = \frac{a\mb I_2 - \vec{w} \cdot \vec{\sigma}}{a^2 - |\vec{w}|^2}, 
\end{aligned}
\end{equation}
where $\vec{w} \times \vec{y}$ is the cross product of two vectors. We have the following explicit calculation for BKM metric: 
\medskip
\begin{lemma}\label{lemma: qubit BKM metric}
For the traceless Hermitian operator $X$ and density operator $\rho$ given by \eqref{eqn: qubit representation},
\begin{equation}\label{eqn: qubit BKM metric}
\begin{aligned}
   g_{\rho}(X) & = 4|\vec{y}|^2 \int_1^{\infty} \frac{u^2 + |\vec{w}|^2 \cos 2\theta}{(u^2- |\vec{w}|^2 )^2} du \\
    & = 2|\vec{y}|^2 \bigg(\frac{1+\cos 2\theta}{ 1- |\vec w|^2} + \frac{1 - \cos 2\theta}{2|\vec w|} \ln \frac{1 + |\vec w|}{1 - |\vec w|} \bigg). 
\end{aligned}
\end{equation}
where $\theta$ is the angle between $\vec{y}$ and $\vec w$.
\end{lemma}
\begin{proof}[Proof of Lemma~\ref{lemma: qubit BKM metric}]
    Recall that $X= \vec{y} \cdot \vec{\sigma}$ and $\rho = \frac{1}{2} (\mb I_2 + \vec{w}\cdot \vec{\sigma})$, use the definition of BKM metric, we have 
    \begin{align*}
         \hspace{0.5cm} g_{\rho}(X)&= \int_0^{\infty} \tr((\vec{y} \cdot \vec{\sigma}) \big( \frac{1}{2} (\mb I_2 + \vec{w}\cdot \vec{\sigma}) + u\mb I_2\big)^{-1} (\vec{y} \cdot \vec{\sigma}) \big( \frac{1}{2} (\mb I_2 + \vec{w}\cdot \vec{\sigma}) + u\mb I_2\big)^{-1} )du \\
         & = 4 \int_0^{\infty} \tr((\vec{y} \cdot \vec{\sigma}) \big( (2u+1)\mb I_2 + \vec{w}\cdot \vec{\sigma} \big)^{-1} (\vec{y} \cdot \vec{\sigma}) \big( (2u+1)\mb I_2 + \vec{w}\cdot \vec{\sigma} \big)^{-1} )du \\
         & = 2 \int_1^{\infty} \tr((\vec{y} \cdot \vec{\sigma}) \big( u\mb I_2 + \vec{w}\cdot \vec{\sigma} \big)^{-1} (\vec{y} \cdot \vec{\sigma}) \big( u\mb I_2 + \vec{w}\cdot \vec{\sigma} \big)^{-1} )du.
    \end{align*}
    Then using the Product rule and Inverse rule in \eqref{eqn: qubit representation calculation}, for any $u > 1$,  we have 
    \begin{align*}
        (\vec{y} \cdot \vec{\sigma}) \big( u\mb I_2 + \vec{w}\cdot \vec{\sigma} \big)^{-1} = \frac{(\vec{y} \cdot \vec{\sigma}) \big(u\mb I_2 - \vec{w}\cdot \vec{\sigma}\big)}{u^2 - |\vec w|^2} = \frac{- (\vec w \cdot \vec y) \mb I_2 + (u \vec y + i \vec w \times \vec y) \cdot \vec \sigma}{u^2 - |\vec w|^2}
    \end{align*}
    thus using the Product rule again,
    \begin{align*}
        & \tr((\vec{y} \cdot \vec{\sigma}) \big( u\mb I_2 + \vec{w}\cdot \vec{\sigma} \big)^{-1} (\vec{y} \cdot \vec{\sigma}) \big( u\mb I_2 + \vec{w}\cdot \vec{\sigma} \big)^{-1} )\\ = &\frac{\tr( \big(- (\vec w \cdot \vec y) \mb I_2 + (u \vec y + i \vec w \times \vec y) \cdot \vec \sigma\big)^2 )}{(u^2 - |\vec w|^2)^2} \\
         = &2\frac{|\vec w \cdot \vec y|^2 + (u \vec y + i \vec w \times \vec y) \cdot (u \vec y + i \vec w \times \vec y)}{(u^2 - |\vec w|^2)^2} \\
         = &2\frac{u^2 |\vec y|^2 + |\vec w \cdot \vec y|^2 - |\vec w \times \vec y|^2 }{(u^2 - |\vec w|^2)^2}. 
    \end{align*}
    Plugging it back to the integral, we have 
    \begin{align*}
        g_{\rho}(X)& = 2\int_1^{\infty} \tr((\vec{y} \cdot \vec{\sigma}) \big( u\mb I_2 + \vec{w}\cdot \vec{\sigma} \big)^{-1} (\vec{y} \cdot \vec{\sigma}) \big( u\mb I_2 + \vec{w}\cdot \vec{\sigma} \big)^{-1} )du \\
        & = 4 \int_1^{\infty} \frac{u^2 |\vec y|^2 + |\vec w \cdot \vec y|^2 - |\vec w \times \vec y|^2 }{(u^2 - |\vec w|^2)^2}du \\ 
        & = 4|\vec{y}|^2 \int_1^{\infty} \frac{u^2 + |\vec{w}|^2 \cos 2\theta}{(u^2- |\vec{w}|^2 )^2} du.
    \end{align*}
    To compute the above integral, note that for $|\vec w|<1$, the following holds:
    \begin{align*}
       & \int_1^{\infty} \frac{u^2}{(u^2- |\vec{w}|^2 )^2} du = \frac{1}{2} \big(\frac{1}{1-|\vec w|^2} - \frac{1}{2|\vec w|} \ln \frac{1 - |\vec w|}{1 + |\vec w|}\big),\\
       & \int_1^{\infty} \frac{1}{(u^2- |\vec{w}|^2 )^2} du = \frac{1}{2|\vec w|^2} \big(\frac{1}{1-|\vec w|^2} + \frac{1}{2|\vec w|} \ln \frac{1 - |\vec w|}{1 + |\vec w|}\big).
    \end{align*}
    Therefore, by some simple algebra, we conclude the proof by showing 
    \begin{align*}
        4|\vec{y}|^2 \int_1^{\infty} \frac{u^2 + |\vec{w}|^2 \cos 2\theta}{(u^2- |\vec{w}|^2 )^2} du = 2|\vec{y}|^2 \bigg(\frac{1+\cos 2\theta}{ 1- |\vec w|^2} + \frac{1 - \cos 2\theta}{2|\vec w|} \ln \frac{1 + |\vec w|}{1 - |\vec w|} \bigg).
    \end{align*}
\end{proof}

Any qubit linear map $\mc N: \mb M_2 \to \mb M_2$ has a one-to-one correspondence to a $4\times4$ matrix $\mc T_{\mc N}$ in the basis of Pauli operators: 
\begin{equation}
    \mc N(c_0\mb I_2 + c_1 \sigma_x + c_2 \sigma_y + c_3 \sigma_z) = c_0'\mb I_2 + c_1' \sigma_x + c_2' \sigma_y + c_3' \sigma_z, \quad \vec{c}' = \mc T_{\mc N} \vec{c}.
\end{equation}
If $\mc N$ is trace-preserving, we must have $c_0 = c_0'$ thus $\mc T_{\mc N}$ has the form 
\begin{equation}
   \mc T_{\mc N}= \begin{pmatrix}
        1 & 0 & 0 & 0\\
        t_1 & a_{11} & a_{12} & a_{13}  \\
        t_2 & a_{21} & a_{22} & a_{23}  \\
        t_3 & a_{31} & a_{32} & a_{33}  
    \end{pmatrix}.
\end{equation}
If $\mc N$ is Hermitian-preserving, it is clear that all the elements of $\mc T_{\mc N}$ are real. Denote 
\begin{equation}
    T = \begin{pmatrix}
         a_{11} & a_{12} & a_{13}  \\
         a_{21} & a_{22} & a_{23}  \\
         a_{31} & a_{32} & a_{33}  
    \end{pmatrix} \in \mb M_3(\mb R),\quad  \vec{t} = \begin{pmatrix}
         t_1 \\
        t_2 \\
        t_3 
    \end{pmatrix} \in \mb R^3.
\end{equation}
For any $\rho = \frac{1}{2} (\mb I_2 + \vec{w}\cdot \vec{\sigma})$, $\mc N(\rho)$ can be represented as
\begin{equation}\label{eqn: qubit channel representation}
    \mc N(\rho) = \frac{1}{2} (\mb I_2 + (T\vec{w} + \vec{t})\cdot \vec{\sigma}).
\end{equation}
We refer the reader to \cite{ruskai02} for a complete analysis on the pair ($T,\vec{t}$) such that $\mc N$ is a quantum channel. Here we only remark that if $\mc N$ is positive, then $\forall \vec{w} \in \mb R^3$ with $|\vec{w}| \le 1$, we have $|T\vec{w} + \vec{t}| \le 1$. 

Given $X = \vec{y} \cdot \vec{\sigma}$ and $\rho = \frac{1}{2}(\mb I_2 + \vec{w}\cdot \vec{\sigma})$ with $ |\vec w|\le 1$, we denote 
\begin{equation}\label{eqn:notation of y and w}
    \vec{y}_{\mc N} = T \vec y,\quad \vec{w}_{\mc N} = T \vec w + \vec t.
\end{equation}
Using Lemma \ref{lemma: qubit BKM metric}, we have 
\begin{align}
    g_{\mc N(\rho)} (\mc  N(X)) & = 2 |\vec{y}_{\mc N}|^2 \bigg(\frac{1+\cos 2\theta_{\mc N}}{ 1- |\vec{w}_{\mc N}|^2} + \frac{1 - \cos 2\theta_{\mc N}}{2|\vec{w}_{\mc N}|} \ln \frac{1 + |\vec{w}_{\mc N}|}{1 - |\vec{w}_{\mc N}|} \bigg) \\
    & = \frac{4|\vec{y}_{\mc N}|^2}{ 1- |\vec{w}_{\mc N}|^2} \bigg(\cos^2 \theta_{\mc N} + \sin^2 \theta_{\mc N} f(|\vec{w}_{\mc N}|)\bigg) \label{eqn:BKM channel},
\end{align}
where $\theta_{\mc N}$ is the angle between $\vec{y}_{\mc N}$ and $\vec{w}_{\mc N}$, and the function $f$ is 
\begin{equation}\label{eqn:auxiliary function f}
    f(x):= \frac{1 -x^2}{2x} \ln \frac{1 + x}{1 - x},\quad x\in [0,1].
\end{equation}
Note that  $f(x) > 0$ for any $x \in [0,1)$ and $f(1) = 0$. When $x \to 1-$, 
\begin{equation}\label{eqn: convergence rate for auxiliary function}
    f(x) \sim -(1-x^2)\ln(1-x^2).
\end{equation}
 See Figure \ref{fig:Auxiliary function} for a plot of this function.
\begin{figure}[ht]
    \centering\includegraphics[width=.5\textwidth]{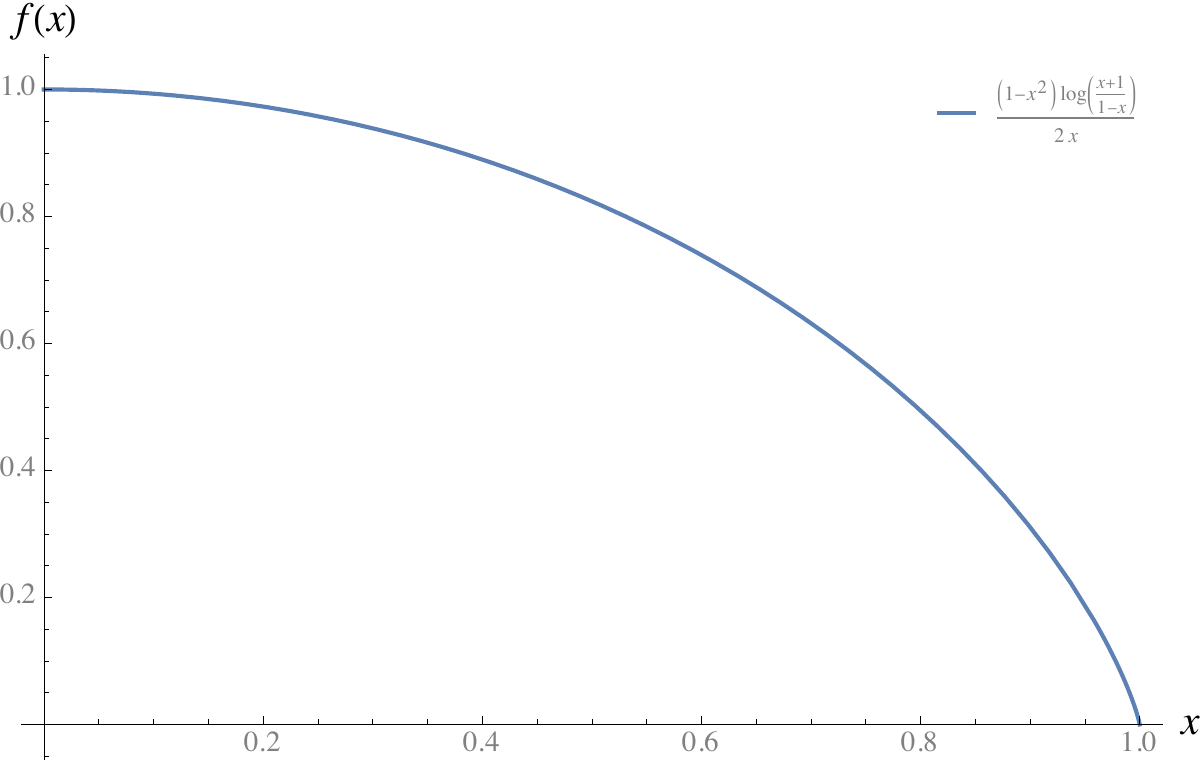}
    \caption{\added{Illustration of the function defined in~\eqref{eqn:auxiliary function f}.}}
    \label{fig:Auxiliary function}
\end{figure}

A concrete estimate for $\eta_{\mc N, \mc M}$ and $\widecheck{\eta}_{\mc N, \mc M}$ can be given directly from \eqref{eqn:BKM channel}:
\medskip
\begin{lemma}\label{lemma:qubit case general result}
Suppose $\mc N$ and $\mc M$ are two qubit trace-preserving and positive maps determined by $(T_1, \vec{t}_1)$ and $(T_2,\vec{t}_2)$, i.e., 
for any $\rho = \frac{1}{2} (\mb I_2 + \vec{w}\cdot \vec{\sigma})$, 
\begin{equation}
    \mc N(\rho) = \frac{1}{2} (\mb I_2 + (T_1\vec{w} + \vec{t}_1)\cdot \vec{\sigma}),\quad \mc M(\rho) = \frac{1}{2} (\mb I_2 + (T_2\vec{w} + \vec{t}_2)\cdot \vec{\sigma}).
\end{equation}
Suppose the following conditions hold: 
\begin{enumerate}
    \item There exist universal constants $c_1> c_2>0$, such that 
    \begin{equation}\label{ineqn:comparison of traceless}
        c_2 |T_1 \vec y| \le |T_2 \vec y| \le c_1 |T_1 \vec y|,\ \forall \vec{y}\in \mb R^3.
    \end{equation}
    \item There exist universal constants $c_3 > c_4>0$, such that for any $\vec{w}$ with $|\vec w|\le 1$, we have 
    \begin{equation}\label{eqn:two channel intermediate step 1}
      c_4 (1 -  |T_1\vec{w} + \vec{t}_1|^2) \le 1 -  |T_2\vec{w} + \vec{t}_2|^2 \le c_3 (1 -  |T_1\vec{w} + \vec{t}_1|^2).
    \end{equation}
    \item There exist universal constants $c_5>c_6>0$, such that for any $\vec{w}$ with $|\vec w|\le 1$ and $\vec y \in \mb R^3$, we have
    \begin{equation}\label{eqn:two channel intermediate step 2}
       c_6 \wt f_1(\vec w, \vec y) \le \wt f_2(\vec w, \vec y) \le c_5 \wt f_1(\vec w, \vec y),
    \end{equation}
    where $\wt f_i(\vec w, \vec y), i = 1,2$ is defined by 
    \begin{equation}
        \wt f_i(\vec w, \vec y) = \cos^2 \theta_i + \sin^2 \theta_i f(|T_i \vec w + \vec{t}_i|),\quad \theta_i = \angle (T_i \vec w + \vec{t}_i, T_i \vec y).
    \end{equation}
\end{enumerate}
Then we have 
\begin{equation}
    \eta_{\mc N, \mc M} \le \frac{c_3}{c_2^2 c_6},\quad \widecheck{\eta}_{\mc N, \mc M} \ge \frac{c_4}{c_1^2 c_5}.
\end{equation}
\end{lemma}
\begin{proof}
The proof follows directly by recalling the expression \eqref{eqn:BKM channel} and estimating the ratio using \eqref{ineqn:comparison of traceless}, \eqref{eqn:two channel intermediate step 1} and \eqref{eqn:two channel intermediate step 2}.
\end{proof}

\begin{remark}{\rm 
    Note that a sufficient condition for \eqref{eqn:two channel intermediate step 1} is
    \begin{align*}
        \mc N^{-1}(\mc P(\mb C^2)) = \mc M^{-1}(\mc P(\mb C^2)),
    \end{align*}
    where $\mc P(\mb C^2)$ is the set of pure qubit states and $\mc N^{-1}(\cdot)$ denotes the pre-image. In fact, if the pre-images of pure states are not the same, then one of the BKM metric can be infinity while the other one is finite. In \cite{HR15}, $\eta_{\mc N,id}$ is explicitly given as $\|T\|^2$ when $\mc N$ is unital. Nevertheless, in our case, the matrices $T_{\mathcal{N}}$ and $T_{\mathcal{M}}$, representing $\mathcal{N}$ and $\mathcal{M}$, may not be simultaneously diagonalizable, making the derivation of an explicit formula more challenging. For conjectured closed-form expressions in specific cases, we refer interested readers to Section~\ref{sec:conclusion and open problems}. 
}
\end{remark}

\section{Examples of channels with non-zero relative expansion coefficients}
\label{sec: examples}
\added{In this section, using the general methods presented in Section \ref{sec: relative contraction and expansion}, we show that the relative expansion coefficient can be strictly positive for pairs of depolarizing channels, pairs of generalized dephasing channels and pairs of qubit amplitude damping channels.} Note that the relative expansion coefficient is non-trivial only if for any states $\rho,\sigma$,
\begin{align*}
    \text{supp} (\mc N(\rho )) \subseteq \text{supp} (\mc N(\sigma)) \Longrightarrow \text{supp} (\mc M(\rho )) \subseteq \text{supp} (\mc M(\sigma)),
\end{align*}
otherwise $D(\mc N(\rho)\|\mc N(\sigma))$ is finite while $D(\mc M(\rho)\|\mc M(\sigma))$ is infinite. 

\subsection{Depolarizing channels} 
For $p\in[0,1]$, a depolarizing channel is defined by
\begin{equation}\label{def: depolarizing}
    \mc D_p: \mb M_d\to \mb M_d\ ,\  \mc D_p(\rho) = (1-p) \rho + \frac{p}{d}\tr(\rho) I_d.
\end{equation}
Our first example of a positive relative expansion coefficient compares two $d$-dimensional depolarizing channels $(\mc D_{p_1},\mc D_{p_2}$). For any such pair with $0< p_2 < p_1 <1$, we show that $\eta_{\mc D_{p_2},\mc D_{p_1}}< \infty$ and $\widecheck{\eta}_{\mc D_{p_1},\mc D_{p_2}}>0$.
For $p_2<p_1$,  as $\mc D_{p_1}= \mc D_{\frac{p_1 - p_2}{1- p_2}} \circ \mc D_{p_2}$, our problem can be seen as a reverse-type data processing inequality restricted on the output states of $\mc D_{p_2}$.
\medskip
\begin{prop}\label{main:depolarizing channel}
For any parameters $0< p_2 < p_1 <1$ and any density operators $\rho, \sigma$, we have 
\begin{equation}\label{main result:depolarizing}
   \left(\frac{1-p_1}{1-p_2}\right)^2\frac{p_2}{p_1} \le \frac{D(\mc D_{p_1}(\rho)\|\mc D_{p_1}(\sigma))}{D(\mc D_{p_2}(\rho)\|\mc D_{p_2}(\sigma))} \le \left(\frac{1-p_1}{1-p_2}\right)^2\frac{1 - \frac{d-1}{d}p_2}{1 - \frac{d-1}{d}p_1}.
\end{equation}
\end{prop}
\begin{proof}Note that for any density $\rho$ and $\sigma$, $\mc D_p(\rho-\sigma)=(1-p)(\rho-\sigma)$. Then for any $\omega$
\begin{equation} \label{eq:depolarizing helpful eq}
    g_{\mc D_{p}(\omega)} (\mc D_{p}(\rho-\sigma)) = (1-p)^2  g_{\mc D_{p}(\omega)} (\rho-\sigma). 
\end{equation}
Moreover, we have 
\begin{equation}
    \frac{1 - \frac{d-1}{d}p_1}{1 - \frac{d-1}{d}p_2} \mc D_{p_2}(\omega) \le \mc D_{p_1}(\omega) \le \frac{p_1}{p_2}\mc D_{p_2}(\omega),
\end{equation}
where the upper and lower bound are given by the supremum and infimum of the function 
\begin{equation}
    h(\lambda) = \frac{(1-p_1)\lambda + \frac{p_1}{d}}{(1-p_2)\lambda + \frac{p_2}{d}},\quad \lambda \in [0,1].
\end{equation}
By applying Lemma~\ref{lemma: comparison}, we have 
\begin{equation}
\begin{aligned}
       \left(\frac{1-p_1}{1-p_2}\right)^2\frac{p_2}{p_1} \le \frac{g_{\mc D_{p_1}(\omega)} (\mc D_{p_1}(\rho-\sigma))}{g_{\mc D_{p_2}(\omega)} (\mc D_{p_2}(\rho-\sigma))} \le \left(\frac{1-p_1}{1-p_2}\right)^2\frac{1 - \frac{d-1}{d}p_2}{1 - \frac{d-1}{d}p_1},
\end{aligned}
\end{equation}
which implies the conclusion via Lemma~\ref{lemma: criteria general}. 
\end{proof}
For qubit case, we can give an explicit expression:
\medskip
\begin{prop}\label{prop: qubit depolarizing}
For two qubit depolarizing channels, we have
\begin{align*}
       \widecheck{\eta}_{\mc D_{p_1},\mc D_{p_2}} = \left(\frac{1-p_1}{1-p_2}\right)^2 \frac{p_2(2-p_2)}{p_1(2-p_1)},
    \end{align*}
    \begin{align*}
       {\eta}_{\mc D_{p_1},\mc D_{p_2}}=\left(\frac{1-p_1}{1-p_2}\right)^2.
    \end{align*}
\end{prop}

\begin{proof}
For any unitary $U$ and any $p\in [0,1]$, we have that $U\mc D_p(\rho) U^{\dagger}=\mc D_p(U\rho U^{\dagger})$, and $D(\mc D_p (\rho) ||\mc D_p (\sigma) )=D(U\mc D_p (\rho) U^{\dagger}||U\mc D_p (\sigma) U^{\dagger})=D(\mc D_p (U \rho U^{\dagger}) ||\mc D_p (U \sigma U^{\dagger}) )$. Taking $U$ to be the (conjugate of the) unitary that diagonalizes $\sigma$, i.e. let $\sigma= U^{\dagger} \begin{pmatrix} 1-\lambda & 0 \\0 & \lambda \end{pmatrix} U$, we can thus restrict ourselves to the case when $\sigma$ is diagonal.

Noting that \eqref{eq:depolarizing helpful eq} still holds, it remains to compute $ \frac{g_{\mc D_{p_1}(\sigma)}(X)}{g_{\mc D_{p_2}(\sigma)(X)}}$. We can use the explicit expression for the BKM metric, for example, \cite[Eq.~16]{AE11}%Lemma~\ref{lemma: qubit BKM metric}
; for a Hermitian matrix $X=\begin{pmatrix}
    x&z\\z^* & -x
\end{pmatrix}$ and diagonal $\sigma$ with eigenvalues $1-\lambda, \lambda$, this becomes $ g_{\sigma}(X)  = (\frac{1}{1 - \lambda} +\frac{1}{\lambda}  ) x^2$. Inserting $\mc D_{p_1}(\sigma), \mc D_{p_2}(\sigma)$ and taking the quotient, we obtain:
\[\frac{g_{\mc D_{p_1}(\sigma)}(X)}{g_{\mc D_{p_2}(\sigma)(X)}} = h(\lambda)\]
with
\begin{equation}
    h(\lambda) = \frac{(2(1-p_1)\lambda +p_1)}{(2(1-p_2)\lambda +p_2)} \frac{(2-2(1-p_1)\lambda -p_1)}{(2-2(1-p_2)\lambda -p_2)},\quad \lambda \in [0,1].
\end{equation}

The supremum of this function is achieved at $\lambda=1/2$, giving $h(1/2)=1$, which corresponds to selecting $\sigma$ as a maximally mixed state. The infimum is taken at $\lambda\rightarrow 0$ (or $\lambda\rightarrow 1$) where it evaluates to $h(1)=\frac{p_2(2-p_2)}{p_1(2-p_1)}$, which corresponds to selecting $\sigma$ as a pure state. Thus, in total, we have:
\begin{equation}
\begin{aligned}
       \left(\frac{1-p_1}{1-p_2}\right)^2 \frac{p_2(2-p_2)}{p_1(2-p_1)}\le \frac{g_{\mc D_{p_1}(\rho_t)} (\mc D_{p_1}(\rho-\sigma))}{g_{\mc D_{p_2}(\rho_t)} (\mc D_{p_2}(\rho-\sigma))} \le \left(\frac{1-p_1}{1-p_2}\right)^2,
\end{aligned}
\end{equation}
 and the upper and lower bound can be achieved.
\end{proof}

\begin{remark}{\rm
    Our upper bound from Proposition~\ref{main:depolarizing channel} implies an upper bound on the contraction coefficient of the depolarizing channel for arbitrary dimensions. In fact, letting $p_2 \to 0$, the upper bound is \begin{align*}
       \eta_{\mc D_p} \le \frac{(1-p)^2}{1 - \frac{d-1}{d}p} < 1- p. 
    \end{align*}
    For the qubit depolarizing channel, it is known that $\eta_{\mc D_p} = (1-p)^2$, see \cite{HR15, HRS22}, which we recover in Proposition~\ref{prop: qubit depolarizing}. This also illustrates that our upper bound from Proposition~\ref{main:depolarizing channel} is not sharp for $d=2$. 
    
    \added{As another point of comparison, when the second state $\sigma$ is the maximally mixed state $I_d / d$, \cite{munch2024intertwining} employ a curvature bound to show that $$D(\mc D_p(\rho) \| I_d / d )\le (1-p)^{1+\frac{1}{d}} D(\rho \| I_d / d )$$ for any $\rho$, extending the result of \cite{MSW16}. More explicitly, \cite{MSW16,munch2024intertwining} studies the local entropy contraction constant when the second state $\omega$ is the fixed point state $I_d / d$
    \[
\eta_{\omega}(\Phi)
=\sup_{\rho}
\frac{D\bigl(\Phi(\rho)\,\|\,\Phi(\omega)\bigr)}
     {D\bigl(\rho\,\|\,\omega\bigr)},
\quad
\omega = \frac{I_d}{d}.
\]
This is potentially smaller than our contraction coefficient
\[
\eta(\Phi)
=\sup_{\rho,\sigma}
\frac{D\bigl(\Phi(\rho)\,\|\,\Phi(\sigma)\bigr)}
     {D\bigl(\rho\,\|\,\sigma\bigr)}.
\]}

    }
\end{remark}

\added{A channel is called strictly positive if it maps any state to a state with full support. For strictly positive channels, we have the following generalization of Proposition \ref{main:depolarizing channel}, }which can be derived from Lemma~\ref{lemma: criteria general} and Lemma~\ref{lemma: comparison}: 
\medskip
\begin{prop}Suppose there exists constants $0<\lambda_{min}<\lambda_{max}<\infty $ such that
     $0<\lambda_{min} I\leq \mc M(\rho) \leq \lambda_{max} I$ for any state $\rho$, then
   \begin{equation}
        \frac{1}{\lambda_{max}}  \|\mc M(\rho-\sigma)\|_2^2 \le D(\mc M(\rho)\|\mc M(\sigma)) \le  \frac{1}{\lambda_{min}} \|\mc M(\rho-\sigma)\|_2^2.
   \end{equation}
\end{prop}
\added{For strictly positive channels, the existence of  $0<\lambda_{min}<\lambda_{max}<\infty $ is guaranteed by the compactness of state space.
Using the above result, we can get a reverse-type data processing inequality for strictly positive channels, which includes depolarizing channels as a special case. For more properties of strictly positive channels, we refer the reader to \cite{Sanz_2010} and references therein. Note that this class of channels does not include generalized dephasing channels or amplitude damping channels, as their output state can be singular.}

\subsection{Generalized dephasing channels}\label{sec: example dephasing}

Another interesting class of quantum channels are quantum dephasing channels which model the loss of coherence (off-diagonal entries of the density matrix) without changing the populations (diagonal elements).

For a $d$-dimensional quantum system $\mc H$, the generalized dephasing channel $\Phi_{\Gamma} : \mb B(\mc H) \to \mb B(\mc H)$ is defined as 
\begin{equation}\label{def: general dephasing}
    \Phi_{\Gamma}(\rho) = \Gamma \odot \rho := \sum_{0\le i,j \le d-1} \Gamma_{ij} \rho_{ij} \ket{i}\bra{j},\quad \rho = \sum_{0\le i,j \le d-1} \rho_{ij} \ket{i}\bra{j},
\end{equation}
where $\Gamma \in \mb B(\mc H)$ such that 
\begin{equation}\label{eqn: dephasing matrix}
    \Gamma_{ij} \in [0,1], \quad \Gamma_{ii} = 1,\ 0\le i,j \le d-1.
\end{equation}
Note that the Choi–Jamio\l{}kowski operator of $\Phi_{\Gamma}$ is 
\begin{equation}
    \mc C_{\Phi_{\Gamma}} = \sum_{i,j = 0}^{d-1} \Gamma_{ij} \ketbra{ii}{jj},
\end{equation}
thus $\Phi_{\Gamma}$ is a quantum channel if and only if $\Gamma$ is positive semidefinite and $\Gamma_{ii}=1$ for all $i$. 

The diagonal entries of a quantum state remain unchanged when a dephasing channel is applied; thus, if we restrict $\rho, \sigma$ to be diagonal operators, we always have 
\begin{equation}
    D(\Phi_{\Gamma}(\rho)\|\Phi_{\Gamma}(\sigma)) = D(\rho \| \sigma) 
\end{equation}
which implies $\eta_{\Phi_{\Gamma}} = 1$, and similarly $\eta_{\Phi,\Phi'}= 1$ for two dephasing channels $\Phi,\Phi'$. For the relative expansion coefficient, using Lemma~\ref{lemma: comparison dephasing}, we show that $\widecheck{\eta}_{\Phi_{\Gamma'},\Phi_{\Gamma}}>0$ for certain positive semidefinite $\Gamma, \Gamma' \in \mb B(\mc H)$. 
\medskip
\begin{prop}\label{proposition: generalized dephasing}
    Let $\Gamma = (\Gamma_{ij}), \Gamma' = (\Gamma_{ij}') \in \mb B(\mc H)$ be positive semidefinite matrix satisfying \eqref{eqn: dephasing matrix}. Suppose there exists $\varepsilon \in (0,\frac{1}{2})$ such that 
    \begin{itemize}
        \item $(1-\varepsilon) \Gamma \le \Gamma' \le (1+\varepsilon)\Gamma$. 
        \item $\widehat{\Gamma} = (\widehat \Gamma_{ij})_{0\le i,j \le d-1}$ is positive semidefinite where  
        \begin{align}
            \widehat \Gamma_{ij} :=
            \begin{cases}
        0,\quad &\text{ if }\quad \Gamma_{ij}' = 0, \\
        \frac{\Gamma_{ij}' - (1-\varepsilon)\Gamma_{ij}}{\varepsilon \Gamma_{ij}'} , &\text{ if }\quad \Gamma_{ij}' > 0.
    \end{cases}\label{eq:def}
        \end{align}
    \end{itemize}
    Then we have 
    \begin{equation}
        \widecheck{\eta}_{\Phi_{\Gamma'},\Phi_{\Gamma}} \ge \frac{(1-2\varepsilon)(1-\varepsilon)}{(1 + 2\varepsilon)(1+\varepsilon)}. 
    \end{equation}
\end{prop}
\begin{proof}
    We identify $\Phi_{\Gamma'}, \Phi_{\Gamma}$ as $\mc M, \mc N$ respectively and verify the assumptions in order to apply Lemma~\ref{lemma: comparison dephasing}. By definition, we have
    \begin{equation}
       \Phi_{\Gamma'} = (1-\varepsilon) \Phi_{\Gamma} + \varepsilon \Phi_{\wt \Gamma},
    \end{equation}
    where $\wt \Gamma_{ij} = \Gamma_{ij} + \frac{\Gamma_{ij}'- \Gamma_{ij}}{\varepsilon}$. $\wt \Gamma = (\wt \Gamma_{ij})_{0\le i,j \le d-1}$ is positive semidefinite by the assumption that $(1-\varepsilon) \Gamma \le \Gamma'$. It remains to show that 
    \begin{enumerate}
        \item There exists a quantum channel $\mc D$ such that $\mc D \circ \Phi_{\Gamma'} = \Phi_{\wt \Gamma}.$
        \item There exists a universal constant $c>0$ such that for any density operator $\sigma$, $\mc D (\Phi_{\Gamma'}(\sigma)) \le c \Phi_{\Gamma'}(\sigma)$.
    \end{enumerate}
    
    For the first argument, we define the generalized dephasing channel $\Phi_{\widehat \Gamma}$ with $\widehat \Gamma$ defined as in \eqref{eq:def}.
    By direct calculation, we have
    \begin{align*}
        \Phi_{\widehat \Gamma} \circ \Phi_{\Gamma'} = \Phi_{\wt \Gamma}. 
    \end{align*}
    By assumption, $\widehat{\Gamma}$ is positive semidefinite and $\widehat{\Gamma}_{ii}=1$, hence
    $\Phi_{\widehat \Gamma}$ is a quantum channel. We choose this channel to be $\mc D$ such that$\mc D \equiv \widehat{\Gamma}$ in the first condition (1).
    
     We will now show that the second condition (2) holds for this choice of $\mc D$. Noting that $\mc D \circ \Phi_{\Gamma'} = \Phi_{\wt \Gamma}$, we have
    \begin{align*}
        \mc D \circ \Phi_{\Gamma'} = \Phi_{\wt \Gamma} = \frac{\Phi_{\Gamma'} - (1-\varepsilon) \Phi_{\Gamma}}{\varepsilon} \le_{cp} \frac{(1+\varepsilon) \Phi_{\Gamma} - (1-\varepsilon) \Phi_{\Gamma}}{\varepsilon} = 2 \Phi_{\Gamma} \le_{cp} \frac{2}{1-\varepsilon} \Phi_{\Gamma'}.
    \end{align*}
    Thus, we can choose $c = \frac{2}{1-\varepsilon}$ in the condition (2).
    
    Finally, we apply Lemma~\ref{lemma: comparison dephasing} with $\omega = \Phi_{p'}(\sigma)$ for any density operator $\sigma$ and $c = \frac{2}{1-\varepsilon}$, and apply Lemma~\ref{lemma: comparison} with $\omega = \Phi_{\Gamma'}(\sigma) \le (1+\varepsilon)\Phi_{\Gamma}(\sigma)$, and we obtain
    \begin{equation}
    \begin{aligned}
    g_{\Phi_{\Gamma'}(\sigma)}(\Phi_{\Gamma'}(X))
    & \ge \frac{(1-2\varepsilon)(1-\varepsilon)}{1+ 2\varepsilon} g_{\Phi_{\Gamma'}(\sigma)}(\Phi_{\Gamma}(X))\\
        & \ge \frac{(1-2\varepsilon)(1-\varepsilon)}{(1+ 2\varepsilon)(1+\varepsilon)} g_{\Phi_{\Gamma}(\sigma)}(\Phi_{\Gamma}(X)),
    \end{aligned}
    \end{equation}
    which implies that $\widecheck{\eta}_{\Phi_{\Gamma'},\Phi_{\Gamma}} \ge \frac{(1-2\varepsilon)(1-\varepsilon)}{(1 + 2\varepsilon)(1+\varepsilon)}$ via Lemma~\ref{lemma: criteria general}.
\end{proof}
\begin{example}\label{example:dephasing}{\rm
    We illustrate our result for the qubit case. In this case, the matrix $\Gamma_p = \begin{pmatrix}
        1 & 1-p \\
        1-p & 1
    \end{pmatrix}$ is determined by a single parameter $p\in [0,2]$ and we denote $\Phi_p = \Phi_{\Gamma_p}$. For \begin{equation} \label{eq: qubit dephasing condition}
        0< p < p' \le (1+\varepsilon)p, \quad \varepsilon\in (0,\frac{1}{2}),
    \end{equation}
    it is easy to verify that both assumptions in Proposition~\ref{proposition: generalized dephasing} hold for qubit dephasing channels $ \Phi_{p'}$ and $\Phi_{p}$, and we thus have $ \widecheck{\eta}_{\Phi_{p'},\Phi_{p}} \ge \frac{(1-2\varepsilon)(1-\varepsilon)}{(1 + 2\varepsilon)(1+\varepsilon)}$ for any pair of channels with $p$ and $p'$ fulfilling \eqref{eq: qubit dephasing condition}, i.e. channels where $p$ and $p'$ are close. }
    % \begin{align*}
    %     \widecheck{\eta}_{\Phi_{p'},\Phi_{p}} \ge \frac{(1-2\varepsilon)(1-\varepsilon)}{(1 + 2\varepsilon)(1+\varepsilon)}. 
    % \end{align*}
\end{example}
To prove that the relative expansion coefficient for qubit dephasing channels is non-zero for arbitrary $p$ and $p'$, we use the following elementary inequality: 
\begin{equation}\label{ineqn:elementary lower bound}
    \inf_{x,y \ge 0} \frac{ax + by}{cx+dy} \ge \min\{\frac{a}{c}, \frac{b}{d}\},\quad a,b,c,d \ge 0.
\end{equation}
%Using the explicit calculation for the qubit case in Section \ref{sec:qubit calculation}, we have the following:
\begin{prop}\label{main:dephasing}
    For any $p_1, p_2 \in (0,1)$, we have 
    \begin{equation}
        \widecheck{\eta}_{\Phi_{p_1},\Phi_{p_2}} > 0.
    \end{equation}
\end{prop}
\begin{proof}
    Using Lemma \ref{lemma: criteria general}, we only need to prove that for two dephasing channels $\Phi_{p_1}$ and $\Phi_{p_2}$ with $0< p_2 < p_1 < 2$ there exists a constant $c(p_1,p_2)>0$, such that for any traceless $X = \vec{y} \cdot \vec{\sigma}$ and $\rho = \frac{1}{2}(\mb I_2 + \vec{w}\cdot \vec{\sigma})$,  we have
\begin{equation}\label{dephasing: goal}
     g_{\Phi_{p_1}(\sigma)}(\Phi_{p_1}(X)) \rangle \ge c(p_1, p_2)  g_{\Phi_{p_2}(\sigma)}(\Phi_{p_2}(X)).
\end{equation}
Note that this is proved for $p_1,p_2$ being close in Example \ref{example:dephasing}. To prove the general case, note that for any $p\in (0,2)$, the qubit representation of $\Phi_p$ as in \eqref{eqn: qubit channel representation} is given by 
\begin{equation}
    \Phi_{p}(\rho) = \frac{1}{2}(\mb I_2 + T_p \vec{w} \cdot \vec{\sigma}),\quad T_p = \text{diag}(1-p,1-p, 1)
\end{equation}
Denote $\vec{y}_p = T_p \vec y, \ \vec{w}_p = T_p \vec w$ and $\theta_p = \angle(\vec{y}_p, \vec{w}_p)$ as the angle between $\vec{y}_p$ and $\vec{w}_p$. We show \eqref{dephasing: goal} by applying Lemma \ref{lemma:qubit case general result}. To be more specific, we verify that
\begin{enumerate}
    \item \label{dephasing: 1} There exist universal constants $c_1> c_2>0$, such that 
    \begin{equation*}
        c_2 |\vec{y}_{p_1}| \le |\vec{y}_{p_2}| \le c_1 |\vec{y}_{p_1}|,\ \forall \vec{y}\in \mb R^3.
    \end{equation*}
    \item \label{dephasing: 2}There exist universal constants $c_3 > c_4>0$, such that for any $\vec{w}$ with $|\vec w|\le 1$, we have 
    \begin{equation*}
      c_4 (1 -  |\vec{w}_{p_1}|^2) \le 1 -  |\vec{w}_{p_2}|^2 \le c_3 (1 -  |\vec{w}_{p_1}|^2)
    \end{equation*}
    \item \label{dephasing: 3}There exist universal constants $c_5>c_6>0$, such that for any $\vec{w}$ with $|\vec w|\le 1$ and $\vec y \in \mb R^3$, we have
    \begin{equation*}
       c_6\le \frac{\cos^2 \theta_{p_1} + \sin^2 \theta_{p_1} f(|\vec{w_{p_1}}|)}{\cos^2 \theta_{p_2} + \sin^2 \theta_{p_2} f(|\vec{w_{p_2}}|) }\le c_5.
    \end{equation*}
\end{enumerate}
\eqref{dephasing: 1} follows directly from the simple form of $T_p = \text{diag}(1-p,1-p, 1)$. For \eqref{dephasing: 2}, we compute
\begin{align}\label{eqn:calculation 1 dephasing}
    1 - |\vec{w}_p |^2 = 1 - \big((1-p)^2(w_1^2 + w_2^2) + w_3^2\big) = (1-|\vec{w}|^2) + p(2-p)(w_1^2 + w_2^2), 
\end{align}
therefore, for any $ \vec w$, 
\begin{align*}
    \frac{1 - |\vec{w}_{p_1}|^2}{1 - |\vec{w}_{p_2}|^2 } & = \frac{(1-|\vec{w}|^2) + p_1(2-p_1)(w_1^2 + w_2^2)}{(1-|\vec{w}|^2) + p_2(2-p_2)(w_1^2 + w_2^2)} \ge \min\{1, \frac{p_1(2-p_1)}{p_2(2-p_2)}\} > 0.
\end{align*}
The hardest part is to show \eqref{dephasing: 3}. Denote $\mc B \subseteq \mb R^3$ as the unit ball, we define a function $g_p : \mc B \times \mb R^3 \to \mb R$ as 
\begin{align}
      \wt f_p(\vec w, \vec y) & := \cos^2 \theta_{p} + \sin^2 \theta_{p} f(|\vec{w_{p}}|), \quad \text{ where }\vec{y}_p = T_p \vec y, \ \vec{w}_p = T_p \vec w, \theta_p = \angle(\vec{y}_p, \vec{w}_p) \label{eqn:g_p:1}\end{align}
      Then one has \begin{align}
         \wt f_p(\vec w, \vec y) & = \cos^2 \theta_{p} (1 - f(|\vec{w_{p}}|))+ f(|\vec{w_{p}}|)\label{eqn:g_p:2} \\
       & = \frac{|\vec{w_{p}} \cdot \vec{y_{p}}|^2}{|\vec{w_{p}}|^2|\vec{y_{p}}|^2} (1 - f(|\vec{w_{p}}|))+ f(|\vec{w_{p}}|) \label{eqn:g_p:3}.
\end{align}
First we note that 
\begin{equation}
    |\vec{w_{p}}| = 1 \iff \vec w = \pm e_3,\quad e_3 = (0,0,1)^T
\end{equation}
Therefore for any $\varepsilon>0$ small, if $\vec w \in \mc B(e_3,\varepsilon)^c \cap \mc B(-e_3,\varepsilon)^c$, using the continuity of $f$, see Figure \ref{fig:Auxiliary function}, there exists a universal constant $c(\varepsilon,p)>0$ such that $f(|\vec{w_{p}}|) \ge c(\varepsilon, p)$, which implies that 
\begin{align*}
    c(\varepsilon,p) \le \cos^2 \theta_{p} + \sin^2 \theta_{p} f(|\vec{w_{p}}|) \le 1.
\end{align*}
Then for $\vec w \in \mc B(e_3,\varepsilon)^c \cap \mc B(-e_3,\varepsilon)^c$, we have 
\begin{align*}
    c(p_1,\varepsilon)\le \frac{\wt f_{p_1}(\vec w, \vec y)}{\wt f_{p_2}(\vec w, \vec y)} \le \frac{1}{c(p_2,\varepsilon)}.
\end{align*}
It remains to show that around the singular points(in this case they are $\pm e_3$), the ratio is lower bounded away from zero. To be more specific, we need to show 
\begin{align*}
    \liminf_{\vec w \to \pm e_3} \inf_{\vec{y}} \frac{\wt f_{p_1}(\vec w, \vec y)}{\wt f_{p_2}(\vec w, \vec y)} = \liminf_{\vec w \to \pm e_3} \inf_{\vec{y}} \frac{\cos^2 \theta_{p_1} + \sin^2 \theta_{p_1} f(|\vec{w_{p_1}}|)}{\cos^2 \theta_{p_2} + \sin^2 \theta_{p_2} f(|\vec{w_{p_2}}|) } > 0. 
\end{align*}
Note that in this case, the elementary lower bound \eqref{ineqn:elementary lower bound} does not work since 
\begin{align*}
    \inf_{\vec{y}}\frac{\cos^2 \theta_{p_1}}{\cos^2 \theta_{p_2}} = 0, \quad \forall \vec w \in \mc B.
\end{align*}
The key idea to show a lower bound is that when $\vec w \to \pm e_3$, if $\cos^2 \theta_{p_i}$ converges to zero, then it tends to zero faster than $f(|\vec{w_{p_1}}|) \sim -(1-|\vec{w_{p_1}}|^2) \ln(1 - |\vec{w_{p_1}}|^2)$ thus a lower bound can still be derived. 

Using \eqref{eqn:g_p:3}, we have 
\begin{align*}
    \liminf_{\vec w \to \pm e_3} \inf_{\vec{y}} \frac{\wt f_{p_1}(\vec w, \vec y)}{\wt f_{p_2}(\vec w, \vec y)}
    & = \liminf_{\vec w \to \pm e_3} \inf_{\vec{y}} \frac{\big|\vec{w_{p_1}} \cdot \vec{y_{p_1}} / |\vec{y_{p_1}}| \big|^2(1 - f(|\vec{w_{p_1}}|))+ |\vec{w_{p_1}}|^2 f(|\vec{w_{p_1}}|)}{\big|\vec{w_{p_2}} \cdot \vec{y_{p_2}} / |\vec{y_{p_2}}| \big|^2(1 - f(|\vec{w_{p_2}}|))+ |\vec{w_{p_2}}|^2 f(|\vec{w_{p_2}}|)}\cdot \frac{|\vec{w_{p_2}}|^2}{|\vec{w_{p_1}}|^2} \\
    & = \liminf_{\vec w \to \pm e_3} \frac{\big|\vec{w_{p_1}} \cdot \vec{y_{p_1}}(\vec w) / |\vec{y_{p_1}}(\vec w)| \big|^2 (1 - f(|\vec{w_{p_1}}|)) + |\vec{w_{p_1}}|^2 f(|\vec{w_{p_1}}|)}{\big|\vec{w_{p_2}} \cdot \vec{y_{p_2}}(\vec w) / |\vec{y_{p_2}}(\vec w)| \big|^2(1 - f(|\vec{w_{p_2}}|))+ |\vec{w_{p_2}}|^2 f(|\vec{w_{p_2}}|)},
\end{align*}
where for each $\vec w \neq \pm e_3$, we denote $\vec{y}(\vec w)$ as 
\begin{equation}
    \vec{y}(\vec w) = \arg\min \frac{\big|\vec{w_{p_1}} \cdot \vec{y_{p_1}} / |\vec{y_{p_1}}| \big|^2(1 - f(|\vec{w_{p_1}}|))+ |\vec{w_{p_1}}|^2 f(|\vec{w_{p_1}}|)}{\big|\vec{w_{p_2}} \cdot \vec{y_{p_2}} / |\vec{y_{p_2}}| \big|^2(1 - f(|\vec{w_{p_2}}|))+ |\vec{w_{p_2}}|^2 f(|\vec{w_{p_2}}|)}, 
\end{equation}
and $\vec{y_{p_i}}(\vec w) := T_{p_i}\vec{y}(\vec w),\ i = 1,2$. Note that the existence of $\vec{y}(\vec w)$ follows from the fact that infimum of a continuous function over a compact set is always achieved. By linearity, we can assume $|\vec y|=1$.  For any $\vec w$ and  $\vec y$
\begin{align*}
    & |\vec{w}_{p_2} \cdot \vec{y}_{p_2}|^2 = |(1-p_2)^2(y_1 w_1 + y_2 w_2) + y_3 w_3|^2 \\
    & = |(1-p_1)^2(y_1 w_1 + y_2 w_2) + y_3 w_3 + \big((1-p_2)^2-(1-p_1)^2\big)(y_1 w_1 + y_2 w_2) |^2 \\
    & \le 2\bigg(|\vec{w}_{p_1} \cdot \vec{y}_{p_1}|^2 + \big((1-p_2)^2-(1-p_1)^2\big)^2|y_1 w_1 + y_2 w_2 |^2 \bigg) \\
    & \le 2\bigg(|\vec{w}_{p_1} \cdot \vec{y}_{p_1}|^2 + \big((2 - p_1-p_2)(p_1 - p_2)\big)^2(y_1^2 + y_2^2)(w_1^2 + w_2^2)  \bigg).
\end{align*}
Recall that $1 - |\vec{w}_p |^2 = 1 - \big((1-p)^2(w_1^2 + w_2^2) + w_3^2\big) = (1-|\vec{w}|^2) + p(2-p)(w_1^2 + w_2^2)$. For $c_1(p_1,p_2):= \frac{(2-p_1-p_2)^2(p_1-p_2)^2}{(1-p_2)^2p_2(2-p_2)}$, we have
\begin{align*}
    \frac{\big((2 - p_1-p_2)(p_1 - p_2)\big)^2(y_1^2 + y_2^2)(w_1^2 + w_2^2)}{|\vec{y}_{p_2}|^2} \le c_1(p_1,p_2) (1 - |\vec{w}_{p_2}|^2) \le c_1(p_1,p_2) f(|\vec{w}_{p_2}|).
\end{align*}
Therefore, using $|\vec{w}_{p_2} \cdot \vec{y}_{p_2}|^2 \le 2 \big(|\vec{w}_{p_1} \cdot \vec{y}_{p_1}|^2 + |\vec{y}_{p_2}|^2 c_1(p_1,p_2) f(|\vec{w}_{p_2}|) \big)$, we have 
\begin{align*}
   & \hspace{0.5cm} \liminf_{\vec w \to \pm e_3} \frac{\big|\vec{w_{p_1}} \cdot \vec{y_{p_1}}(\vec w) / |\vec{y_{p_1}}(\vec w)| \big|^2 (1 - f(|\vec{w_{p_1}}|)) + |\vec{w_{p_1}}|^2 f(|\vec{w_{p_1}}|)}{\big|\vec{w_{p_2}} \cdot \vec{y_{p_2}}(\vec w) / |\vec{y_{p_2}}(\vec w)| \big|^2(1 - f(|\vec{w_{p_2}}|))+ |\vec{w_{p_2}}|^2 f(|\vec{w_{p_2}}|)} \\
   & \ge  \liminf_{\vec w \to \pm e_3} \frac{\big|\vec{w_{p_1}} \cdot \vec{y_{p_1}}(\vec w) / |\vec{y_{p_1}}(\vec w)| \big|^2 (1 - f(|\vec{w_{p_1}}|)) + |\vec{w_{p_1}}|^2 f(|\vec{w_{p_1}}|)}{2\big|\vec{w_{p_1}} \cdot \vec{y_{p_1}}(\vec w) / |\vec{y_{p_2}}(\vec w)| \big|^2(1 - f(|\vec{w_{p_2}}|))+ (|\vec{w_{p_2}}|^2 + 2c_1(p_1,p_2))f(|\vec{w_{p_2}}|)} \\
   & \underset{\eqref{ineqn:elementary lower bound}}{\ge}  \min\big\{\frac{(1-p_2)^2}{2(1-p_1)^2}, \frac{1}{1 + c_1(p_1,p_2)} \liminf_{\vec w \to \pm e_3} \frac{f(|\vec{w_{p_1}}|)}{f(|\vec{w_{p_2}}|)}\big\} \ge \min\big\{\frac{(1-p_2)^2}{2(1-p_1)^2}, \frac{1}{1 + c_1(p_1,p_2)}\big\}>0,
\end{align*}
where in the last inequality, we used
\begin{align*}
    & \liminf_{\vec w \to \pm e_3} \frac{f(|\vec{w_{p_1}}|)}{f(|\vec{w_{p_2}}|)} =  \liminf_{\vec w \to \pm e_3} \frac{(1- |T_{p_1} \vec{w}|^2) \ln (1-|T_{p_1} \vec{w}|^2)}{(1- |T_{p_2} \vec{w}|^2) \ln (1-|T_{p_2} \vec{w}|^2)} \ge \min\{1, \frac{p_1(2-p_2)}{p_2(2-p_2)}\} = 1.
\end{align*}
\end{proof}
Before we proceed to the next example, we remark here that we have to take $\liminf$ in the above proof, since the limit may not exist. 

\subsection{Amplitude damping channels}
We now study the relative expansion coefficient for qubit amplitude damping channels. For $\gamma \in (0,1)$, we define the amplitude damping channel $\mc A_{\gamma}$ as 
\begin{equation}\label{def:amplitude damping channel}
    \mc A_{\gamma}\begin{pmatrix}
        \rho_{00} & \rho_{01} \\
        \rho_{10} & \rho_{11}
    \end{pmatrix} = \begin{pmatrix}
        \rho_{00}+ \gamma \rho_{11} & \sqrt{1-\gamma}\rho_{01} \\
        \sqrt{1-\gamma} \rho_{10} & (1-\gamma)\rho_{11}
    \end{pmatrix}.
\end{equation}
Note that amplitude damping channels does not satisfy the cp order comparison in order to apply Lemma~\ref{lemma: comparison dephasing}. In fact, %using the Choi–Jamio\l{}kowski operator, we see that
for any $\gamma_1, \gamma_2 \in (0,1)$ and any $c>0$, $\mc A_{\gamma_1} - c \mc A_{\gamma_2}$ is not completely positive. Therefore, the techniques in the previous sections do not apply here. Instead, we will use the explicit calculation of the BKM-metric  $g_{\sigma}(X)$ for a qubit density operator $\sigma$ from Lemma~\ref{lemma: qubit BKM metric} and Lemma~\ref{lemma:qubit case general result} to derive the positivity of the relative expansion coefficient of two amplitude damping channels: 
\medskip
\begin{prop}\label{main:amplitude damping}
    For any $\gamma_1,\gamma_2 \in (0,1)$, we have 
    \begin{equation}
        \widecheck{\eta}_{\mc A_{\gamma_1},\mc A_{\gamma_2}} > 0.
    \end{equation}
\end{prop}
\begin{figure}[h]
    \centering\includegraphics[width=0.8\textwidth]{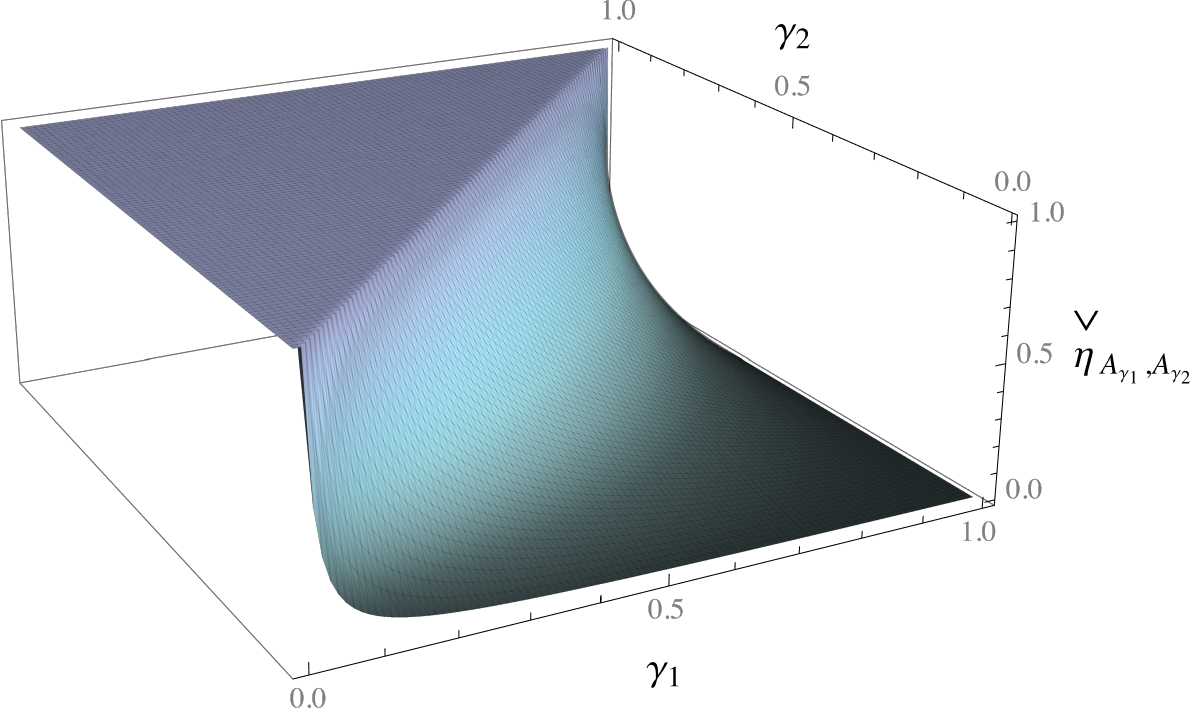}
    \caption{A plot of $\widecheck{\eta}_{\mc A_{\gamma_1},\mc A_{\gamma_2}}$ for $\gamma_1,\gamma_2 \in (0,1)$.}
    \label{fig:relative expansion amplitude damping}
\end{figure}
\begin{proof}
%[Proof of Proposition~\ref{main:amplitude damping} for non-diagonal states $\sigma$]
We aim to show that for any $0<\gamma_1 <\gamma_2<1$, there exists a constant $c(\gamma_1,\gamma_2)>0$ such that for any density operator $\rho$ and traceless Hermitian operator $X$, we have  
\begin{equation}\label{amplitude: goal}
    g_{\mc A_{\gamma_1}(\rho)}(\mc A_{\gamma_1}(X)) \ge c(\gamma_1,\gamma_2) g_{\mc A_{\gamma_2}(\rho)}(\mc A_{\gamma_2}(X)).
\end{equation}
Suppose $X = \vec{y} \cdot \vec{\sigma},\ \rho = \frac{1}{2} (\mb I_2 + \vec{w}\cdot \vec{\sigma})$, the qubit representation of $\mc A_{\gamma}$ as in \eqref{eqn: qubit channel representation} is given by 
\begin{equation}
    \mc A_{\gamma}(\rho) = \frac{1}{2} (\mb I_2 + (T_{\gamma} \vec w + \vec{t}_{\gamma})\cdot \vec{\sigma}), \quad T_{\gamma} = diag(\sqrt{1-\gamma},\sqrt{1-\gamma}, 1-\gamma),\ \vec{t}_{\gamma} = (0,0,\gamma)^T.
\end{equation}
Denote $\vec{w}_{\gamma} = T_{\gamma} \vec w + \vec{t}_{\gamma},\ \vec{y}_{\gamma} = T_{\gamma} \vec y$. Then via Lemma \ref{lemma: qubit BKM metric}, we have
\begin{equation}\label{eqn: BKM amplitude}
 \begin{aligned}
 g_{\mc A_{\gamma}(\rho)}(\mc A_{\gamma}(X)) =  4\frac{|\vec{y_{\gamma}}|^2}{1- |\vec{w_{\gamma}}|^2} \bigg(\cos^2 \theta_{\gamma} + \sin^2 \theta_{\gamma} f(|\vec{w}_{\gamma}|)\bigg),
\end{aligned}   
\end{equation}
where $\theta_{\gamma}$ is the angle between $\vec{w_{\gamma}}$ and $\vec{y_{\gamma}}$, and $f$ is defined in \eqref{eqn:auxiliary function f}. Similar to the proof of dephasing channels, We show \eqref{amplitude: goal} by applying Lemma \ref{lemma:qubit case general result}. To be more specific, we need to show \\

\begin{enumerate}
    \item \label{amplitude: 1} There exist universal constants $c_1> c_2>0$, such that 
    \begin{equation*}
        c_2 |\vec{y}_{\gamma_1}| \le |\vec{y}_{\gamma_2}| \le c_1 |\vec{y}_{\gamma_1}|,\ \forall \vec{y}\in \mb R^3.
    \end{equation*}
    \item \label{amplitude: 2}There exist universal constants $c_3 > c_4>0$, such that for any $\vec{w}$ with $|\vec w|\le 1$, we have 
    \begin{equation*}
      c_4 (1 -  |\vec{w}_{\gamma_1}|^2) \le 1 -  |\vec{w}_{\gamma_2}|^2 \le c_3 (1 -  |\vec{w}_{\gamma_1}|^2).
    \end{equation*}
    \item \label{amplitude: 3}There exist universal constants $c_5>c_6>0$, such that for any $\vec{w}$ with $|\vec w|\le 1$ and $\vec y \in \mb R^3$, we have
    \begin{equation*}
       c_6\le \frac{\cos^2 \theta_{\gamma_1} + \sin^2 \theta_{\gamma_1} f(|\vec{w_{\gamma_1}}|)}{\cos^2 \theta_{\gamma_2} + \sin^2 \theta_{\gamma_2} f(|\vec{w_{\gamma_2}}|) }\le c_5.
    \end{equation*}
\end{enumerate}
Note that \eqref{amplitude: 1} follows directly from the form of $T_{\gamma} $. To show \eqref{amplitude: 2}, recall that $$\vec{w_{\gamma}} = (\sqrt{1-\gamma} w_1, \sqrt{1-\gamma} w_2, (1-\gamma)w_3 + \gamma)^T, $$ we have $1 - |\vec{w_{\gamma}}|^2 = (1-\gamma)(1 - |\vec w|^2) + \gamma(1-\gamma) (w_3 - 1)^2$. Therefore, 
\begin{align*}
    \frac{1 - |\vec{w_{\gamma_1}}|^2}{1 - |\vec{w_{\gamma_2}}|^2} \ge \frac{1-\gamma_1}{1 - \gamma_2},\quad \frac{1 - |\vec{w_{\gamma_1}}|^2}{1 - |\vec{w_{\gamma_2}}|^2} \ge \frac{\gamma_2(1-\gamma_2)}{\gamma_1(1 - \gamma_1)}.
\end{align*}
To show \eqref{amplitude: 3}, we follow the same approach as qubit dephasing channels in Proposition \ref{main:dephasing}. First note that $|\vec{w}_{\gamma}| = 1$ if and only if $\vec{w} = e_3$. Using the same compactness argument, we only need to show \begin{align*}
    & \liminf_{\vec w \to e_3} \inf_{\vec{y}} \frac{\cos^2 \theta_{\gamma_1} + \sin^2 \theta_{\gamma_1} f(|\vec{w_{\gamma_1}}|)}{\cos^2 \theta_{\gamma_2} + \sin^2 \theta_{\gamma_2} f(|\vec{w_{\gamma_2}}|) } \\
    & = \liminf_{\vec w \to e_3} \frac{\big|\vec{w_{\gamma_1}} \cdot \vec{y_{\gamma_1}}(\vec w) / |\vec{y_{\gamma_1}}(\vec w)| \big|^2 (1 - f(|\vec{w_{\gamma_1}}|)) + |\vec{w_{\gamma_1}}|^2 f(|\vec{w_{\gamma_1}}|)}{\big|\vec{w_{\gamma_2}} \cdot \vec{y_{\gamma_2}}(\vec w) / |\vec{y_{\gamma_2}}(\vec w)| \big|^2(1 - f(|\vec{w_{\gamma_2}}|))+ |\vec{w_{\gamma_2}}|^2 f(|\vec{w_{\gamma_2}}|)}> 0,
\end{align*}
where $\vec{y}(\vec w)$ is the minimizer. For any $\gamma \in (0,1)$, by direct calculation, for any $\vec w, \vec y$:
\begin{align*}
    \vec{w_{\gamma}} \cdot \vec{y_{\gamma}} & = (\sqrt{1-\gamma} w_1, \sqrt{1-\gamma} w_2, (1-\gamma)w_3 + \gamma) \cdot (\sqrt{1-\gamma} y_1, \sqrt{1-\gamma} y_2, (1-\gamma) y_3) \\
    & = (1-\gamma) \big( \vec{w} \cdot \vec{y} + \gamma y_3(1-w_3) \big),
\end{align*}
which implies
\begin{align*}
    |\vec{w_{\gamma_2}} \cdot \vec{y_{\gamma_2}}|^2 /|\vec{y_{\gamma_2}}|^2 & = (1 - \gamma_2) \big|\big(\vec{w} \cdot \vec{y} + \gamma_2 y_3 (1-w_3)\big) \big|^2 /|\vec{y_{\gamma_2}}|^2 \\
    & = \big(\frac{1-\gamma_2}{1 - \gamma_1}\big)^2 (1-\gamma_1)^2 \big|\big(\vec{w} \cdot \vec{y} + \gamma_1 y_3 (1-w_3) + (\gamma_2 - \gamma_1) y_3 (1-w_3)\big) \big|^2 /|\vec{y_{\gamma_2}}|^2 \\
    & \le 2 \big(\frac{1-\gamma_2}{1 - \gamma_1}\big)^2 \frac{|\vec{w_{\gamma_1}} \cdot \vec{y_{\gamma_1}}|^2}{|\vec{y_{\gamma_2}}|^2} + 2(1-\gamma_2)^2(\gamma_2 - \gamma_1)^2 \frac{y_3^2 (1-w_3)^2}{|\vec{y_{\gamma_2}}|^2} \\
    & \le 2 \big(\frac{1-\gamma_2}{1 - \gamma_1}\big)^2 \frac{|\vec{w_{\gamma_1}} \cdot \vec{y_{\gamma_1}}|^2}{|\vec{y_{\gamma_2}}|^2} + 2\frac{(\gamma_2 - \gamma_1)^2}{\gamma_2(1-\gamma_2)} (1 -|\vec{w}_{\gamma_2}|^2) \\
    & \le 2 \big(\frac{1-\gamma_2}{1 - \gamma_1}\big)^2 \frac{|\vec{w_{\gamma_1}} \cdot \vec{y_{\gamma_1}}|^2}{|\vec{y_{\gamma_2}}|^2} + 2\frac{(\gamma_2 - \gamma_1)^2}{\gamma_2(1-\gamma_2)} f (|\vec{w}_{\gamma_2}|).
\end{align*}
Therefore, 
\begin{align*}
    & \liminf_{\vec w \to e_3} \frac{\big|\vec{w_{\gamma_1}} \cdot \vec{y_{\gamma_1}}(\vec w) / |\vec{y_{\gamma_1}}(\vec w)| \big|^2 (1 - f(|\vec{w_{\gamma_1}}|)) + |\vec{w_{\gamma_1}}|^2 f(|\vec{w_{\gamma_1}}|)}{\big|\vec{w_{\gamma_2}} \cdot \vec{y_{\gamma_2}}(\vec w) / |\vec{y_{\gamma_2}}(\vec w)| \big|^2(1 - f(|\vec{w_{\gamma_2}}|))+ |\vec{w_{\gamma_2}}|^2 f(|\vec{w_{\gamma_2}}|)} \\
    & \ge \liminf_{\vec w \to e_3} \frac{\big|\vec{w_{\gamma_1}} \cdot \vec{y_{\gamma_1}}(\vec w) / |\vec{y_{\gamma_1}}(\vec w)| \big|^2 (1 - f(|\vec{w_{\gamma_1}}|)) + |\vec{w_{\gamma_1}}|^2 f(|\vec{w_{\gamma_1}}|)}{ 2 \big(\frac{1-\gamma_2}{1 - \gamma_1}\big)^2 |\vec{w_{\gamma_1}} \cdot \vec{y_{\gamma_1}}(\vec{w})|^2 /|\vec{y_{\gamma_2}}(\vec w)|^2 (1 - f(|\vec{w_{\gamma_2}}|)+ (2\frac{(\gamma_2 - \gamma_1)^2}{\gamma_2(1-\gamma_2)} + |\vec{w_{\gamma_2}}|^2) f(|\vec{w_{\gamma_2}}|)} \\
    & \ge \min\{\frac{1}{2}, \frac{1}{2\frac{(\gamma_2 - \gamma_1)^2}{\gamma_2(1-\gamma_2)} + 1} \liminf_{\vec w \to e_3} \frac{f(|\vec{w_{\gamma_1}}|)}{f(|\vec{w_{\gamma_2}}|)}\} \\
    & \ge \frac{(1-\gamma_1)\gamma_2}{2(\gamma_1-\gamma_2)^2 + (1-\gamma_2)\gamma_2}>0,
\end{align*}
where we used the following lower bound for the last inequality:
\begin{align*}
    \liminf_{\vec w \to e_3} \frac{f(|\vec{w_{\gamma_1}}|)}{f(|\vec{w_{\gamma_2}}|)} = \liminf_{\vec w \to e_3} \frac{(1 -|\vec{w_{\gamma_1}}|^2)\ln (1 -|\vec{w_{\gamma_1}}|^2)}{(1 -|\vec{w_{\gamma_2}}|^2)\ln (1 - |\vec{w_{\gamma_2}}|^2) } \ge \frac{1-\gamma_1}{ 1- \gamma_2}. 
\end{align*}
This concludes \eqref{amplitude: 3} and thus finishes the proof.
\end{proof}

\section{Application: less noisy but non-degradable channels}\label{sec: applications}
\added{In this section, we discuss an application of the result in Section \ref{sec: examples}.} Based on the idea of reverse-type data processing inequalities and flag-extension of quantum channels \cite{SS08,SW24}, we construct a family of parameterized quantum channels and show that these channels are less noisy for a certain parameter region. Within this region, the constructed channel is neither degradable nor anti-degradable, giving a way to construct examples of channels that are less noisy but not degradable.

\subsection{Probabilistic mixture of degradable and anti-degradable channels}
Suppose $\mc N$ and $\mc M$ are two degradable channels. Define 
\begin{equation}
    \Psi_{p,\mc N,\mc M}:= p \ketbra{0}{0}\otimes \mc N + (1-p) \ketbra{1}{1}\otimes \mc M^c, 
\end{equation}
which is a probabilistic mixture of degradable and anti-degradable channels. We denote the isometries generating $\mc N$ and $\mc M$ as 
\begin{equation}
    U_{\mc N}: \mc H_A \to \mc H_{B_1} \otimes \mc H_{E_1},\quad U_{\mc M}: \mc H_A \to \mc H_{B_2} \otimes \mc H_{E_2}
\end{equation}
and denote $\mc D_1$ and $\mc D_2$ as the degrading quantum channels respectively, i.e.,
\begin{equation}
    \mc D_1 \circ \mc N = \mc N^c,\quad \mc D_2 \circ \mc M = \mc M^c,
\end{equation}
where $\mc N^c$ and $\mc M^c$ are complementary channels.

A sufficient condition for $\Psi_{p,\mc N,\mc M}$ to be less noisy is given as follows:
\medskip
\begin{prop}\label{main:less noisy}
    Suppose $\mc N, \mc M$ are two degradable channels such that the relative expansion coefficient of $(\mc N,\mc M)$ is positive, i.e., $\widecheck{\eta}_{\mc N,\mc M} >0$. Furthermore, we assume the degrading channel $\mc D_1$ (degrading for $\mc N$) satisfies strong data processing inequality, i.e., $\eta_{\mc D_1}<1$. Then for any $$p \in [\frac{1}{1 + \widecheck{\eta}_{\mc N,\mc M} (1- \eta_{\mc D_1})},1],$$ the quantum channel 
    \begin{align}\label{eqn:flagged mixture structure}
        \Psi_{p,\mc N,\mc M}= p \ketbra{0}{0}\otimes \mc N + (1-p) \ketbra{1}{1}\otimes \mc M^c
    \end{align}
    is less noisy. 
\end{prop}
\begin{proof}
    Our goal is to show that for any classical-quantum state $\rho_{\mc X A} = \sum_{x \in \mc X} p_x \ketbra{x}{x} \otimes \rho_A^x$, we have 
\begin{align*}
    I(\mc X;B) \ge I(\mc X;E), 
\end{align*}
where $\rho_{\mc X B} = \sum_{x \in \mc X} p_x \ketbra{x}{x} \otimes \Psi_{p,\mc N,\mc M}(\rho_A^x)$ and $\rho_{\mc X E} = \sum_{x \in \mc X} p_x \ketbra{x}{x} \otimes \Psi^c_{p,\mc N,\mc M}(\rho_A^x)$. \added{Denoting the isometry $U_{\mc N}: A \to B_1 E_1$ and $U_{\mc M}: A \to B_2 E_2$, we can decompose $\mc H_B = \mc H_{B_1} \oplus \mc H_{E_2},\ \mc H_E = \mc H_{E_1} \oplus \mc H_{B_2}$, and 
\begin{align*}
    & \rho_{\mc X B_1} = \sum_{x \in \mc X} p_x \ketbra{x}{x} \otimes \mc N(\rho_A^x),\quad \rho_{\mc X B_2} = \sum_{x \in \mc X} p_x \ketbra{x}{x} \otimes \mc M(\rho_A^x), \\
    & \rho_{\mc X E_1} = \sum_{x \in \mc X} p_x \ketbra{x}{x} \otimes \mc N^c(\rho_A^x),\quad \rho_{\mc X E_2} = \sum_{x \in \mc X} p_x \ketbra{x}{x} \otimes \mc M^c(\rho_A^x).
\end{align*}}
Noting that the mutual information under convex combination of orthogonal states is additive, we have
\begin{equation}
    I(\mc X;B) - I(\mc X;E) = p (I(\mc X;B_1) - I(\mc X;E_1))- (1-p) (I(\mc X;B_2) - I(\mc X;E_2)).
\end{equation}
Therefore, $I(\mc X;B) - I(\mc X;E) \ge 0$ is equivalent to 
\begin{align}\label{eqn: key inequality of less noisy}
    \frac{I(\mc X;B_1) - I(\mc X;E_1)}{I(\mc X;B_2) - I(\mc X;E_2)} \ge \frac{1-p}{p}.
\end{align}
This holds ture if $p \ge \frac{1}{1 + \widecheck{\eta}_{\mc N,\mc M} (1- \eta_{\mc D_1})}$. In fact, 
\begin{align*}
    \frac{I(\mc X;B_1) - I(\mc X;E_1)}{I(\mc X;B_2) - I(\mc X;E_2)} = \frac{I(\mc X;B_1)}{I(\mc X;B_2)} \left(\frac{1- \frac{I(\mc X;E_1)}{I(\mc X;B_1)} }{1- \frac{I(\mc X;E_2)}{I(\mc X;B_2)}} \right)\ge \frac{I(\mc X;B_1)}{I(\mc X;B_2)} \left(1- \frac{I(\mc X;E_1)}{I(\mc X;B_1)}\right).
\end{align*}
Using \cite[Proposition 2.3]{HRS22}, we see that, for any fixed $\eta > 0$, for any $\rho_{\mc X A}$ with $\tr_{\mc X}(\rho_{\mc X A}) = \sigma_A$, we have $I(\mc X;B_1) \ge  \eta I(\mc X;B_2)$ if and only if, for any $\rho_A$ with $\text{supp}(\rho_A) \subseteq \text{supp}(\sigma_A)$, we have $D(\mc N(\rho_A)\|\mc N(\sigma_A)) \ge \eta D(\mc M(\rho_A)\|\mc M(\sigma_A))$. %, meaning that the relative contraction of relative entropies is equivalent to relative contraction of the mutual information. 
%Choosing $\eta =\widecheck{\eta}_{\mc N,\mc M}$, we thus have
Therefore, we have $\frac{I(\mc X;B_1)}{I(\mc X;B_2)} \ge \widecheck{\eta}_{\mc N,\mc M}$ with the relative expansion coefficient $ \widecheck{\eta}_{\mc N,\mc M}$ defined by \eqref{def:rel-expan}. Similarly, we have $\frac{I(\mc X;E_1)}{I(\mc X;B_1)} \le \eta_{\mc D_1}$ with the contraction coefficient defined by \eqref{def:contraction coefficient}. Therefore, we have 
\begin{align*}
    \frac{I(\mc X;B_1) - I(\mc X;E_1)}{I(\mc X;B_2) - I(\mc X;E_2)} \ge \widecheck{\eta}_{\mc N,\mc M}(1- \eta_{\mc D_1}) \ge \frac{1-p}{p}
\end{align*}
if $p \ge \frac{1}{1 + \widecheck{\eta}_{\mc N,\mc M} (1- \eta_{\mc D_1})}$, which concludes the proof. 
\end{proof}

\subsection{Explicit construction using amplitude damping channels}
For two amplitude damping channels with parameter $\gamma_1,\gamma_2 \in (0,1)$, we define its probabilistic mixture as
\begin{equation}\label{def:target channel}
\begin{aligned}
       & \Psi_{p,\gamma_1,\gamma_2}(\rho) = p \ketbra{0}{0}\otimes \mc A_{\gamma_1}(\rho) + (1-p) \ketbra{1}{1} \otimes \mc A_{\gamma_2}(\rho).
\end{aligned}
\end{equation}
The regions of degradability and anti-degradability are given in \cite[Proposition IV.1]{SW24}, finding that
$\Psi_{p,\gamma_1,\gamma_2}$ is degradable if and only if $(p,\gamma_1, \gamma_2)$ satisfies one of the following conditions:
    \begin{enumerate}
        \item For $p = \frac{1}{2}$: $\gamma_1 + \gamma_2 \le 1$. 
        \item For $p> \frac{1}{2}$: $\gamma_1 + \gamma_2 \le 1$ and $\gamma_1 \le \frac{1}{2}$.
        \item For $p<\frac{1}{2}$: $\gamma_1 + \gamma_2 \le 1$ and $\gamma_2 \le \frac{1}{2}$.
    \end{enumerate}
    On the other hand, $\Psi_{p,\gamma_1,\gamma_2}$ is anti-degradable if and only if $(p,\gamma_1, \gamma_2)$ satisfies one of the following conditions:
    \begin{enumerate}
        \item For $p = \frac{1}{2}$: $\gamma_1 + \gamma_2 \ge 1$. 
        \item For $p> \frac{1}{2}$: $\gamma_1 + \gamma_2 \ge 1$ and $\gamma_1 \ge \frac{1}{2}$.
        \item For $p<\frac{1}{2}$: $\gamma_1 + \gamma_2 \ge 1$ and $\gamma_2 \ge \frac{1}{2}$.
    \end{enumerate}
% \begin{figure}[ht]
%     \centering\includegraphics[width=.8\textwidth]{region_prob_mixture.png}
%     \caption{Degradable and anti-degradable regions for probabilistic mixture of two amplitude damping channels defined in \eqref{def:target channel}.}
%     \label{fig:deg region}
% \end{figure}
\begin{figure}[ht]
    \centering\includegraphics[width=0.5\textwidth]{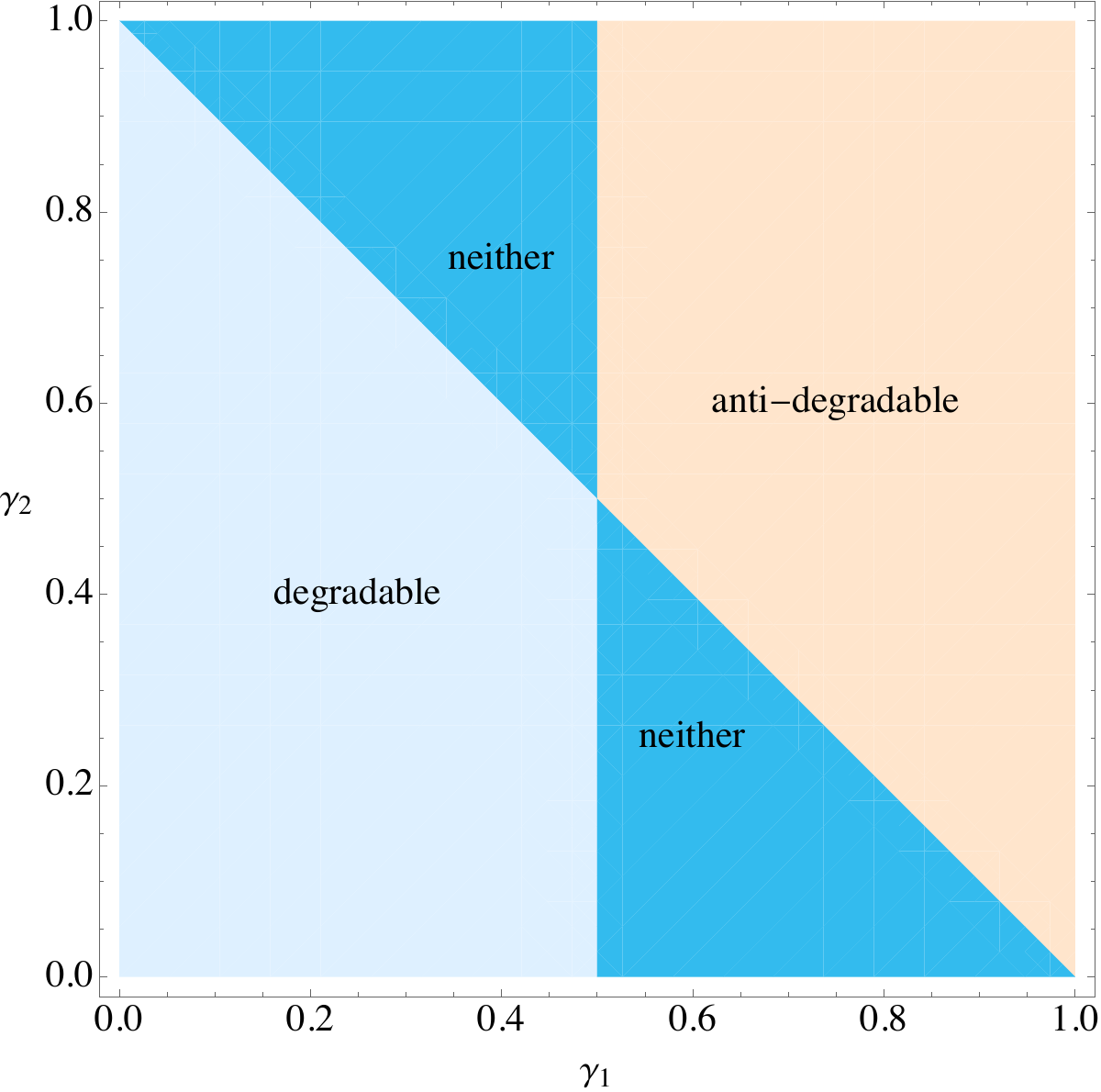} 
    \caption{Degradable and anti-degradable regions for probabilistic mixture of two amplitude damping channels defined in \eqref{def:target channel} in the case $p>\frac{1}{2}$. A plot of the corresponding regions for the case of $p<\frac{1}{2}$ can be found in \cite{SW24}.}
    \label{fig:deg region}
\end{figure}
Recall that for an amplitude damping channel $\mc A_{\gamma}$, its complementary channel is given by an amplitude damping channel $\mc A_{1-\gamma}$. Furthermore, $\mc A_{\gamma}$ is degradable if and only if $\gamma \le \frac{1}{2}$ and the degrading channel $\mc D$ is given by another amplitude damping channel $\mc D = \mc A_{\wt \gamma}$ with damping parameter $\wt \gamma = \frac{1 - 2\gamma}{1-\gamma} \in (0,1)$. Therefore, in order to apply Proposition~\ref{main:less noisy}, we need to show that the contraction coefficient of any amplitude damping channel is strictly less than 1.

%Although it is not explicitly given in the literature, 
We can estimate the contraction coefficient of an amplitude damping channel using the results in \cite{HR15, HT24} which connect the contraction coefficient using relative entropy and the contraction coefficient using trace distance:
\medskip
\begin{lemma}\label{lemma:contraction coefficient entropy and trace}
    For any quantum channel $\mc N$, denote 
    \begin{equation}\label{def:contraction coefficient entropy and trace}
        \eta_{\mc N}:= \sup_{\rho \neq \sigma} \frac{D(\mc N(\rho) || \mc N(\sigma))}{D(\rho || \sigma)}, \quad \eta_{\mc N}^{\tr}:= \sup_{\rho \neq \sigma} \frac{\tr(|\mc N(\rho)- \mc N(\sigma)|)}{\tr(|\rho - \sigma|)}.
    \end{equation}
    Then we have 
    \begin{equation}
        (\eta_{\mc N}^{\tr})^2 \le \eta_{\mc N} \le \eta_{\mc N}^{\tr}.
    \end{equation}
\end{lemma}
\begin{proof}
    The upper bound is given in~\cite[Lemma 4.1]{HT24} using integral representation (see also \cite[Theorem 4.6]{LR99} for a spectral method). The lower bound is given in \cite[Theorem 5.3 \& Theorem 7.1]{HR15}. For the convenience of the reader, we provide a self-contained proof in Appendix~\ref{appendix: proof of contraction entropy and trace}.
\end{proof}
This implies an estimate of the contraction coefficient for amplitude damping channels:
\medskip
\begin{lemma}\label{lemma:contraction amplitude damping}
    For $\gamma \in (0,1)$, we have 
    \begin{equation}
        1-\gamma \le \eta_{\mc A_{\gamma}} \le \sqrt{1-\gamma}.
    \end{equation}
\end{lemma}
\begin{proof}
    Because of Lemma~\ref{lemma:contraction coefficient entropy and trace}, we only need to calculate $\eta_{\mc A_{\gamma}}^{\tr}$. In fact, for any qubit density operators $\rho, \sigma$, we have
    \begin{align*}
        \rho - \sigma = \begin{pmatrix}
            x & z \\
            z^* & -x
        \end{pmatrix}, \quad \tr(|\rho - \sigma|) = 2 \sqrt{x^2 + |z|^2}.
    \end{align*}
    for some $x\in \mb R$, $z\in \mb C$.
    After application of an amplitude damping channel, we have
    \begin{align*}
        \mc A_{\gamma}(\rho - \sigma) = \begin{pmatrix}
            (1-\gamma)x & \sqrt{1-\gamma}z \\
            \sqrt{1-\gamma}z^* & -(1-\gamma)x
        \end{pmatrix},\quad \tr(|\mc A_{\gamma}(\rho - \sigma)|) = 2 \sqrt{(1-\gamma)^2 x^2 + (1-\gamma)|z|^2}.
    \end{align*}
    Therefore, we have
    \begin{align*}
        \frac{\tr(| \mc A_{\gamma}(\rho)- \mc A_{\gamma}(\sigma)|)}{\tr(|\rho - \sigma|)}  = \sqrt{\frac{(1-\gamma)^2 x^2 + (1-\gamma)|z|^2}{x^2 + |z|^2}} \in [1-\gamma, \sqrt{1-\gamma}].
    \end{align*}
    Choosing $x=0$, we achieve $\eta_{\mc A_{\gamma}}^{\tr} =\sqrt{1-\gamma}$. Therefore, using Lemma~\ref{lemma:contraction coefficient entropy and trace}, we get the desired result for $\eta_{\mc A_{\gamma}}$.
\end{proof}

Using Proposition~\ref{main:less noisy} and Lemma~\ref{lemma:contraction amplitude damping}, we can determine the region where the channel $\Psi_{p,\gamma_1,\gamma_2}$ from \eqref{def:target channel} is less noisy:
\medskip
\begin{prop}\label{prop:ampdamp regions}
$\Psi_{p,\gamma_1,\gamma_2}$ is less noisy if 
\begin{itemize}
    \item $\gamma_1 + \gamma_2 >1$ and $\gamma_1 < \frac{1}{2}$, and
%     \begin{align*}
%             p \in [\frac{1}{1 + \widecheck{\eta}_{\mc A_{\eta}, \mc A_{1-p}} (1- { {\frac{1-2\eta}{1-\eta}}})} , 1].
% \end{align*}
\begin{align*}
            p \in \left[\frac{1}{1 + \widecheck{\eta}_{\mc A_{\gamma_1}, \mc A_{1-\gamma_2}} (1- \eta_{\mc A_{\frac{1-2\gamma_1}{1-\gamma_1}}})} , 1\right].
\end{align*}
\item $\gamma_1 + \gamma_2 >1$ and $\gamma_2 < \frac{1}{2}$, and
        \begin{align*}
            p \in \left[0, \frac{\widecheck{\eta}_{\mc A_{\gamma_2}, \mc A_{1-\gamma_1}} (1- \eta_{\mc A_{\frac{1-2\gamma_2}{1-\gamma_2}}})}{1 + \widecheck{\eta}_{\mc A_{\gamma_2}, \mc A_{1-\gamma_1}} (1- \eta_{\mc A_{\frac{1-2\gamma_2}{1-\gamma_2}}})}\right].
        \end{align*}
\end{itemize}
\end{prop}
Recall that $\Psi_{p,\gamma_1,\gamma_2}$ is not degradable for the parameter regions $\gamma_1 + \gamma_2 >1$, $\gamma_1< \frac{1}{2}$ and $ p> \frac{1}{2}$ or $\gamma_1 + \gamma_2 >1$, $\gamma_2 < \frac{1}{2}$ and $p< \frac{1}{2}$, \added{see Figure \ref{fig:deg region} and the conditions above it}. Then we have:
\medskip
\begin{cor}
   There are non-trivial regions of $(\gamma_1,\gamma_2)$ where the channel is $\Psi_{p,\gamma_1,\gamma_2}$ less noisy but not degradable. 
\end{cor}
A concrete example of a less noisy but not degradable channel is $\Psi_{p,\gamma_1,\gamma_2}$ for parameters $p=0.75$, $\gamma_1=0.2$ and $\gamma_2=0.81$. In fact, we obtain a whole parameter region by using our explicit expression for $\widecheck{\eta}_{\mc A_{\gamma_2}, \mc A_{1-\gamma_1}}$ and the upper bound $\eta_{\mc A_{\frac{1-2\gamma_2}{1-\gamma_2}}} \leq \sqrt{1-\frac{1-2\gamma_1}{1-\gamma_1}}$ from Lemma~\ref{lemma:contraction amplitude damping}. Let $p_{min}(\gamma_1,\gamma_2):= \frac{1}{ 1-(1 - \sqrt{1-\frac{1-2\gamma_1}{1-\gamma_1}})\frac{(1-\gamma_1)(1-\gamma_2)}{\gamma_1\gamma_2}}$. Then, $p_{min}(\gamma_1,\gamma_2)\in [\frac{1}{1 + \widecheck{\eta}_{\mc A_{\gamma_1}, \mc A_{1-\gamma_2}} (1- \eta_{\mc A_{\frac{1-2\gamma_1}{1-\gamma_1}}})} , 1]$. In the region $\gamma_1 + \gamma_2 >1$, $\gamma_1 < \frac{1}{2}$, $p\geq p_{min}(\gamma_1,\gamma_2)>\frac{1}{2}$, which is highlighted in Figure~\ref{fig:rainbow}, the channel $\Psi_{p,\gamma_1,\gamma_2}$ is not degradable (by \cite{SW24}) and less noisy (by Proposition~\ref{prop:ampdamp regions}).

\textbf{\begin{figure}[ht]
    \centering\includegraphics[width=0.8
\textwidth]{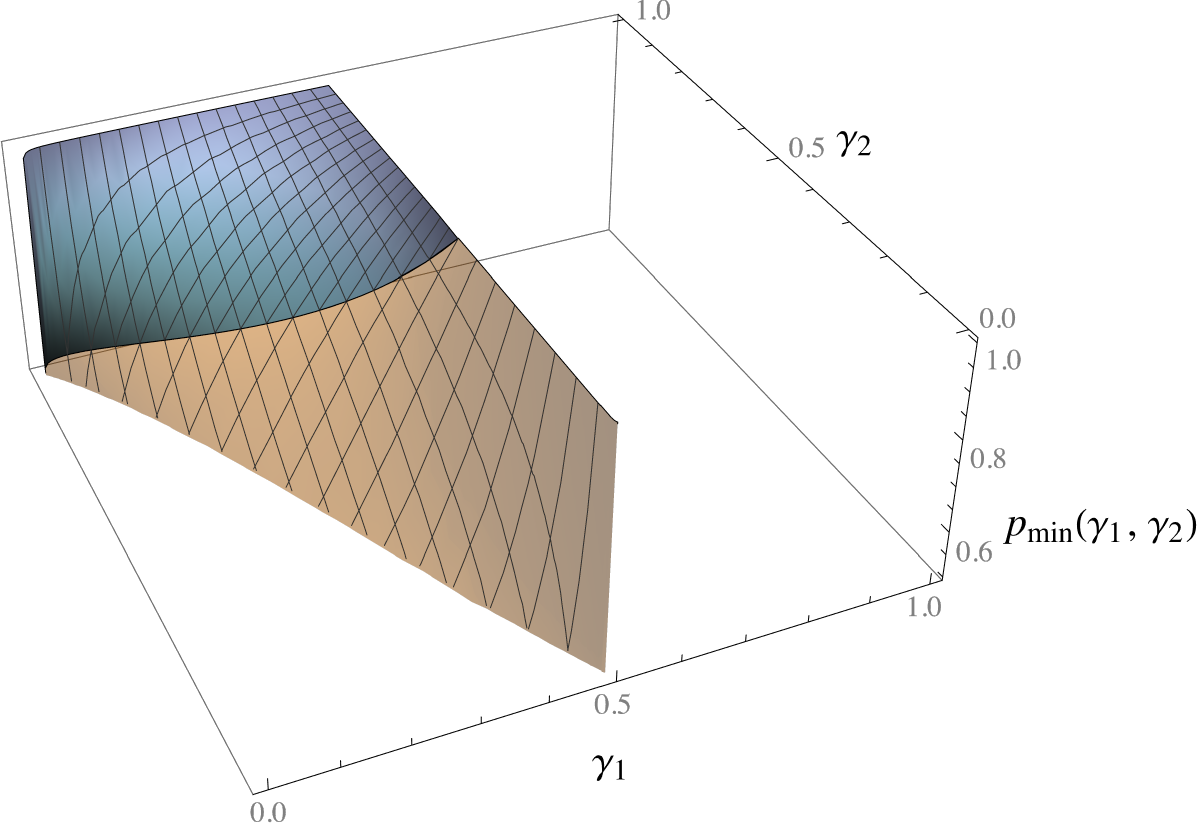} 
    \caption{We plot $p_{min}(\gamma_1,\gamma_2)$, an upper bound on the cutoff probability above which $\Psi_{p, \gamma_1,\gamma_2}$ is less noisy by Proposition~\ref{prop:ampdamp regions}. The highlighted region is the region where $\Psi_{p, \gamma_1,\gamma_2}$ is less noisy but not degradable for any $p\geq p_{min}(\gamma_1,\gamma_2)$. (See the degradability regions in Fig.~\ref{fig:deg region}.)}
    \label{fig:rainbow}
\end{figure}}

 %We conjecture that a complete version of the above is also true thus there are non-trivial regions where the channel is informationally degradable but not degradable.

% \subsection{Exponential decay of relative entropy?}

% By repeated application of a quantum channel, how fast does it go to zero?
% \cite{GFY18}

%\subsection{Private capacity idea?}

\section{Conclusion and Open Problems}\label{sec:conclusion and open problems}

%In this work, we first showed that channels with greater input dimension cannot have non-zero expansion coefficient, hence can not fulfill a reverse data processing inequality. We propose a comparative approach by introducing a relative expansion coefficient, to assess how one channel expands relative entropy compared to another. Then we provided quantitative estimates of relative contraction and expansion coefficients for pairs of quantum channels. These coefficients were based on the Umegaki relative entropy, is a fundamental information measure in quantum information. Based on those new estimates, we provide the first rigorous construction of level-1 less noisy channels which are not degradable. Several further questions can be discussed. First, quantum data processing inequalities, contraction coefficients, and partial orders can also be defined with respect to other information measures, such as quantum $f$-divergences \cite{Petz85,OP93}, which have been explored in prior works \cite{Wilde18,HRS22,HT24,GZB24}. An intriguing direction for future work is to investigate whether our results for relative entropy extend to other quantum divergences. 

In this work, we explored the DPI and reverse DPI of relative entropy through contraction and expansion coefficients of quantum channels. We first showed that channels with greater input dimension than output dimension cannot have non-zero expansion coefficient, and hence can not fulfill a reverse DPI. We studied relative expansion and contraction coefficients and gave quantitative estimates for several important pairs of quantum channels, including in particular pairs of amplitude damping channels.
Based on those new estimates for amplitude damping channels, we provide the first rigorous construction of level-1 less noisy channels which are not degradable. Several interesting open problems remain, some of which we list below.

\textbf{Other information measures.} Quantum DPI, contraction coefficients, and partial orders can also be defined with respect to other information measures, such as quantum $f$-divergences \cite{Petz85,OP93}, which have been explored in prior works \cite{Wilde18,HRS22,HT24,GZB24}. \added{Some of the results established in this work can be generalized to other information measures in a straightforward way. For example, Theorem \ref{main:impossibility} holds for a larger class of functions beyond the $f$-divergence that corresponds to the relative entropy. However, establishing the equivalence between $f$-divergences and their Riemannian metric can be intriguing and we leave it for future work.}

%\subsection*{Complete relative contraction and expansion coefficients}
\textbf{Complete relative contraction and expansion coefficients.}
Another interesting question pertains to the definition of relative contraction and expansion coefficients in terms of the quantum mutual information with fully quantum systems. Using mathematical terminology, it asks whether our inequalites tensorize. Specifically, we propose the following complete versions of these coefficients:
\medskip
\begin{definition}\label{remark:complete version}
We define the complete relative contraction and expansion coefficients as:
\begin{align*}
    \eta_{\mc N,\mc M}^{cb} &= \sup_{\rho_{VA}} \frac{I(V;B_1)}{I(V;B_2)}, \quad 
    \widecheck{\eta}^{cb}_{\mc N,\mc M} = \inf_{\rho_{VA}} \frac{I(V;B_1)}{I(V;B_2)}.
\end{align*}
Here, $\mc N : \mb B(\mc H_A) \to \mb B(\mc H_{B_1})$ and $\mc M : \mb B(\mc H_A) \to \mb B(\mc H_{B_1})$, and $I(V;B)$ denotes the mutual information. 
\end{definition}
\added{The above quantities are non-trivial. In fact, using the joint convexity of the relative entropy, we can easily obtain that $\eta_{\mc D_p}^{cb} \leq 1-p < 1$ for the depolarizing channel as defined in \eqref{def: depolarizing}. Extending techniques from this work, \cite{H2024} and \cite{WBCDT24} to bound complete relative expansion and contraction coefficients remains an open question. In comparison, this definition of complete contraction coefficient $\eta_{\mc D_1,\id}^{cb}$ differs from the one proposed in \cite{HRS22}, which is always equal to 1. Another type of tensorization
\begin{equation}
    \sup_{n\ge 1}\sup_{\rho^n\neq \sigma^n}\frac{D\bigl(\mathcal{N}^{\otimes n}(\rho^n)\,\big\|\,
    \mathcal{N}^{\otimes n}(\sigma^n)\bigr)}{D\bigl(\mathcal{M}^{\otimes n}(\rho^n)\,\big\|\,
    \mathcal{M}^{\otimes n}(\sigma^n)\bigr)}
\end{equation}
is also known not to tensorize well for relative entropy \cite{Cao_2019}. Therefore, the tensorization proposed in Definition \ref{remark:complete version} is a promising tensorization approach. 
}Using the construction outlined in Proposition~\ref{main:less noisy}, for two degradable channels $\mc N$ and $\mc M$, if $\widecheck{\eta}^{cb}_{\mc N,\mc M} > 0$ and $\eta_{\mc D_1,\id}^{cb} < 1$ where $\mc D_1$ is the degrading channel $\mc D_1 \circ \mc N = \mc N^c$, then for any $p \in \left[\frac{1}{1 + \widecheck{\eta}^{cb}_{\mc N,\mc M} (1 - \eta_{\mc D_1,\id}^{cb})}, 1\right]$, the channel
\begin{align*}
    \Psi_{p,\mc N,\mc M} = p \ketbra{0}{0}\otimes \mc N + (1-p) \ketbra{1}{1}\otimes \mc M^c
\end{align*}
is informationally degradable which leads to additivity of quantum capacity.

%\subsection*{Closed-form expressions }\label{app:conjecture formula}
%We discuss how to derive closed-form expressions for the relative contraction and relative expansion coefficients from numerical evidence. 

\textbf{Closed-form expressions.} Numerically computing exact values of the relative expansion coefficient $\widecheck{\eta}_{\mc N,\mc M}$ involves optimizing over pairs of density operators, which is generally computationally expensive. To obtain exact values of $\widecheck{\eta}_{\mc N,\mc M}$ for our examples, we represent a density operator $\rho \in \mb M_d$ through its purification $\ket{\psi} \in \mathbb{C}^{d^2}$, which can be expressed as a unit real vector $v \in \mathbb{R}^{2d^2}$. The function
\[
F(\rho, \sigma) = \frac{D(\mc N(\rho) \|\mc N(\sigma)) }{D(\mc M(\rho) \|\mc M(\sigma))}
\]
is then viewed as a function of two unit real vectors $v_1, v_2 \in \mathbb{R}^{2d^2}$ and optimized using standard numerical schemes such as \texttt{fminunc} in MATLAB. To mitigate the risk of the optimization being trapped in local minima, we randomize the initial values and select the global minimum from multiple runs.

For several of our qubit channel examples, we observe that the optimum is achieved around a certain state with a channel-dependent perturbation. This observation leads us to conjecture closed-form expressions for the relative contraction and relative expansion coefficients of the amplitude damping channel.

For two amplitude damping channels with $\gamma_1, \gamma_2 \in (0,1)$, we observe that the infimum appearing in $\widecheck{\eta}_{\mc A_{\gamma_1}, \mc A_{\gamma_2}}$ is achieved around the state $\ketbra{1}{1}$. Specifically, setting $\rho = \ketbra{1}{1}$ and $\sigma_{\varepsilon} = \varepsilon \ketbra{0}{0} + (1-\varepsilon) \ketbra{1}{1}$, we find that
\[
\widecheck{\eta}'_{\mc A_{\gamma_1}, \mc A_{\gamma_2}} = \lim_{\varepsilon \to 0} \frac{D(\mc A_{\gamma_1}(\rho) \|\mc A_{\gamma_1}(\sigma_{\varepsilon}))}{D(\mc A_{\gamma_2}(\rho) \|\mc A_{\gamma_2}(\sigma_{\varepsilon}))} = \frac{\gamma_2(1-\gamma_1)}{\gamma_1(1-\gamma_2)},
\]
which matches perfectly with the numerical results obtained from the standard procedure. We therefore conjecture that $\widecheck{\eta}'_{\mc A_{\gamma_1}, \mc A_{\gamma_2}} =\widecheck{\eta}_{\mc A_{\gamma_1}, \mc A_{\gamma_2}}$.

For the relative contraction coefficient $\eta_{\mc A_{\gamma_1}, \mc A_{\gamma_2}}$, let $\rho = (1-p) \ketbra{0}{0} + p \ketbra{1}{1}$ and $\sigma_{\varepsilon} = \rho + \varepsilon \sigma_x$, where $\sigma_x$ is the Pauli-X operator. Optimizing over $p \in [0,1]$ and letting $\varepsilon \rightarrow 0$, we can obtain an explicit formula
\[
\eta'_{\mc A_{\gamma_1}, \mc A_{\gamma_2}} = \frac{1-\gamma_1}{1-\gamma_2} \, c(\gamma_1,\gamma_2),
\]
where
\begin{align*}
    c(\gamma_1,\gamma_2) = \max_{p\in [0,1]} \frac{(1 - 2(1-\gamma_2)p) \left[\log(1-(1-\gamma_1)p) - \log((1-\gamma_1)p)\right]}{(1 - 2(1-\gamma_1)p) \left[\log(1-(1-\gamma_2)p) - \log((1-\gamma_2)p)\right]}.
\end{align*}
This formula perfectly matches the numerical results for computing the true value of $\eta_{\mc A_{\gamma_1}, \mc A_{\gamma_2}}$. We therefore conjecture that $\eta'_{\mc A_{\gamma_1}, \mc A_{\gamma_2}}=\eta_{\mc A_{\gamma_1}, \mc A_{\gamma_2}}$.

For two dephasing channels, numerical evidence shows that the optimizers of $\widecheck{\eta}_{\Phi_{p_1}, \Phi_{p_2}}$ for $p_1, p_2 \in (0,1)$ are around a pure state. However, the exact value of $\widecheck{\eta}_{\Phi_{p_1}, \Phi_{p_2}}$ depends on the choice of the pure state, and we do not obtain a closed-form formula based on this method.

\appendix

\section{Proofs}\label{appendix:proof}
\subsection{Proof of Lemma~\ref{lemma: comparison}}
\label{appendix: proof of comparison}
This was proven previously in \cite[Lemma 2.1]{GR22}.
\begin{proof}
    %We provide a proof for the convenience of the reader.
    Using operator anti-monotonicity for the function $f(t) = \frac{1}{t + r}$, i.e., $A^{\dagger} (\rho + rI)^{-1} A \le A^{\dagger}  (c\sigma+ rI)^{-1} A$ for any operator $A$, and by the cyclicity of the trace, we have
    \begin{align*}
        %\langle X, \mc J_{\rho} (X)\rangle
        g_{\rho}(X) & = \int_0^{\infty} \tr\left(X^{\dagger} (\rho + rI)^{-1} X (\rho + rI)^{-1}\right) \dd r \\
        & = \int_0^{\infty} \tr\left((X(\rho + rI)^{-\frac{1}{2}})^{\dagger} (\rho + rI)^{-1} X (\rho + rI)^{-\frac{1}{2}}\right) \dd r \\
        & \le \int_0^{\infty} \tr\left((X(\rho + rI)^{-\frac{1}{2}})^{\dagger} (c \sigma + rI)^{-1} X (\rho + rI)^{-\frac{1}{2}}\right) \dd r \\
        & = \int_0^{\infty} \tr\left((c \sigma + rI)^{-\frac{1}{2}}X (\rho + rI)^{-1} ((c \sigma + rI)^{-\frac{1}{2}}X )^{\dagger}\right) \dd r \\
        & \le \int_0^{\infty} \tr\left((c \sigma + rI)^{-\frac{1}{2}}X (c \sigma + rI)^{-1} ((c \sigma + rI)^{-\frac{1}{2}}X )^{\dagger}\right) \dd r \\
        & = \int_0^{\infty} \tr\left(X^{\dagger} (c \sigma + rI)^{-1} X (c \sigma + rI)^{-1}\right)\dd r \\
        & = \frac{1}{c} g_{\sigma}(X),%\langle X, \mc J_{\sigma} (X)\rangle,
    \end{align*}
    where for the last equality, we used the change of variable $r \mapsto r/c$.
\end{proof}

\subsection{Proof of Lemma~\ref{lemma:contraction coefficient entropy and trace}} \label{appendix: proof of contraction entropy and trace}
This proof is extracted from \cite[Theorem 5.3, Theorem 7.1]{HR15} and \cite[Lemma 4.1]{HT24} . 
\begin{proof}
    \textbf{To prove $(\eta_{\mc N}^{\tr})^2 \le \eta_{\mc N}$}, we show that for any Hermitian traceless operator $X$, there exists a density operator $\sigma$, such that 
    \begin{equation}\label{contraction coefficient entropy and trace: key step 1}
        \frac{(\tr |\mc N(X)|)^2}{(\tr |X|)^2} \le \frac{g_{\mc N (\sigma)}(\mc N(X) )}{g_{\sigma}(X)}.
    \end{equation}
    Then via Lemma \ref{lemma: criteria general}, we conclude the proof by taking the supremum over $X$. To show \eqref{contraction coefficient entropy and trace: key step 1}, we claim that for any Hermitian traceless operator $X$ and density operator $\sigma$, we have 
    \begin{equation}\label{contraction coefficient entropy and trace: key step 2}
        (\tr |X|)^2 \le g_{\sigma}(X).
    \end{equation}
    In fact, given $X$, denote $\mc E_{X}$ as the trace-preserving conditional expectation onto the sub-algebra generated by $X$ (in particular $\mc E_{X}(X) = X$), then using data processing inequality, we have
    \begin{align*}
        g_{\sigma}(X) = \langle X, \mc J_{\sigma}(X)\rangle & \ge \langle \mc E_{X}(X), \mc J_{\mc E_{X}(\sigma)}(\mc E_{X}(X) )\rangle = \langle X, \mc J_{\mc E_{X}(\sigma)}(X)\rangle \\
        & = \tr(X \int_0^{\infty} (\mc E_{X}(\sigma) +r I)^{-1} X (\mc E_{X}(\sigma) +r I)^{-1} dr) \\
        & = \tr(\mc E_{X}(\sigma)^{-1} X^2) = \tr(\mc E_{X}(\sigma)) \cdot \tr(\mc E_{X}(\sigma)^{-1} X^2) \\
        & \ge \big(\tr(\mc E_{X}(\sigma)^{1/2} \mc E_{X}(\sigma)^{-1/2} |X|)\big)^2 = (\tr |X|)^2,
    \end{align*}
    where the integral calculation follows from the fact that $\mc E_{X}(\sigma)$ commutes with $X$ and the last inequality is Cauchy-Schwartz inequality. 

    Then replacing $X$ by $\mc N(X)$ and $\sigma$ by $\mc N(\sigma)$ in \eqref{contraction coefficient entropy and trace: key step 2}, we have 
    \begin{align*}
        (\tr |\mc N(X)|)^2 \le  g_{\sigma}(X).
    \end{align*}
    To compare the denominator in \eqref{contraction coefficient entropy and trace: key step 1}, we choose a special density operator 
    \begin{equation}\label{contraction coefficient entropy and trace: choice of state}
        \sigma = \frac{|X|}{\tr(|X|)},\quad X \neq 0.
    \end{equation}
    Then using the commutativity of $\sigma$ and $X$, we have 
    \begin{align*}
         g_{\sigma}(X)  = \tr(\sigma^{-1} X^2 ) = (\tr |X|)^2. 
    \end{align*}
    In summary, we proved \eqref{contraction coefficient entropy and trace: key step 1} with $\sigma$ given by \eqref{contraction coefficient entropy and trace: choice of state} thus finished the proof of $(\eta_{\mc N}^{\tr})^2 \le \eta_{\mc N}$.  

    \noindent \textbf{To prove $\eta_{\mc N} \le \eta_{\mc N}^{\tr}$}, we use the $L^1$-type integral representation of relative entropy, given by 
    \begin{equation}\label{eqn: L^1 representation}
        D(\rho \|\sigma) = \int_1^{\infty}  \left(\frac{1}{s} E_{s}(\rho \| \sigma) + \frac{1}{s^2} E_{s}(\sigma \| \rho) \right)ds,
    \end{equation}
    where the Hockey-Stick divergence $E_{s}(\rho \| \sigma)$ is given by 
    \begin{equation}
        E_{s}(\rho \| \sigma):= \tr\big( (\rho - s \sigma)_+ \big).
    \end{equation}
    We refer the reader to \cite[Corollary 2.3]{HT24} and \cite[Theorem 6]{Frenkel_2023} for the proof of \eqref{eqn: L^1 representation}. For the Hockey-Stick divergence, we have 
    \begin{equation}\label{eqn: contraction coefficient Hockey-Stick}
        \eta_{s}(\mc N): = \sup_{\rho \neq \sigma} \frac{E_{s}(\mc N(\rho) \| \mc N(\sigma))}{E_{s}(\rho \| \sigma)} = \sup_{\ket{\psi}, \ket{\phi}} E_{s}(\mc N(\ketbra{\psi}{\psi}) \| \mc N(\ketbra{\phi}{\phi})) \le \eta_{\tr}(\mc N),
    \end{equation}
    where the last equality is proved in \cite{Hirche_privacy} and the inequality follows from $E_{s}(\rho \| \sigma) \le E_{1}(\rho \| \sigma)$.
    Then given any $\rho,\sigma$, we have 
    \begin{equation}
    \begin{aligned}
        D(\mc N(\rho) \|\mc N(\sigma) ) & = \int_1^{\infty} \left( \frac{1}{s} E_{s}(\mc N(\rho) \|\mc N(\sigma) ) + \frac{1}{s^2} E_{s}(\mc N(\sigma) \|\mc N(\rho) ) \right) ds \\
        & \le \int_1^{\infty} \left( \frac{1}{s} \eta_{s}(\mc N) E_{s}(\rho \| \sigma)  + \frac{1}{s^2} \eta_{s}(\mc N) E_{s}(\sigma \| \rho)\right) ds \\
        & \le \int_1^{\infty} \left( \frac{1}{s} \eta_{\tr}(\mc N) E_{s}(\rho \| \sigma)  + \frac{1}{s^2} \eta_{\tr}(\mc N) E_{s}(\sigma \| \rho) \right) ds \\
        & = \eta_{\tr}(\mc N) D(\rho \|\sigma ),
    \end{aligned}
    \end{equation}
    which concludes the proof. Note that the first inequality uses the definition of $\eta_{s}(\mc N)$ and for the second inequality, we used \eqref{eqn: contraction coefficient Hockey-Stick}.
\end{proof}

\section*{Acknowledgements} The authors would like to thank Christoph Hirche for helpful comments on the draft and Frederik vom Ende for pointing out the reference \cite{davies1976quantum}. 

\section*{Funding} PB, GS and PW acknowledge funding from the Canada First Research Excellence Fund. LG acknowledges funding by the National Natural Science Foundation of China (grant No. 12401163).

\section*{Data Availability} Data sharing not applicable to this article as no datasets were generated or analyzed during the current study. 

\section*{Conflict of interest} The authors have no conflicts of interest to declare that are relevant to the content of this article.

\bibliography{rdpi}
\end{document}